\documentclass[ejs,noinfoline]{imsart}

\RequirePackage[OT1]{fontenc}
\RequirePackage{amsthm,amsmath,amssymb}
\RequirePackage[numbers]{natbib}
\RequirePackage[colorlinks,citecolor=blue,urlcolor=blue]{hyperref}

\pubyear{2017}
\volume{11}
\issue{1}
\firstpage{2519}
\lastpage{2546}

\makeatletter
\def\input@path{{../sections/}{../}}
\makeatother

\usepackage{graphicx}
\graphicspath{{../images/}{../srcimages}}

\startlocaldefs
\numberwithin{equation}{section}
\theoremstyle{plain}
\newtheorem{theorem}{Theorem}[section]
\newtheorem{definition}{Definition}[section]
\newtheorem{proposition}{Proposition}[section]
\newtheorem{assumption}[theorem]{Assumption}
\newtheorem{lemma}[theorem]{Lemma}
\newtheorem{corollary}[theorem]{Corollary}
\endlocaldefs

\usepackage[plain,noend]{algorithm2e}

\usepackage{graphicx}
\graphicspath{{../images/}}
\usepackage{enumitem}

\usepackage{our_shortcuts}

\begin{document}

\begin{frontmatter}
\title{Optimal~Two-Step Prediction~in~Regression}
\runtitle{Two-Step Prediction}

\begin{aug}

\author{\fnms{Didier} \snm{Ch\'etelat}
\ead[label=e1]{didier.chetelat@hec.ca}}

\address{
Department of Decision Sciences\\
HEC Montr\'eal\\
3000, chemin de la C\^ote-Sainte-Catherine\\
Montr\'eal, Canada\\
\printead{e1}
}

\author{\fnms{Johannes} \snm{Lederer}
\ead[label=e2]{ledererj@uw.edu}}

\address{
Departments of Statistics and Biostatistics\\
University of Washington\\
Box 354322\\
Seattle, WA 98195-4322\\
\printead{e2}
}

\author{\fnms{Joseph} \snm{Salmon}
\ead[label=e3]{joseph.salmon@telecom-paristech.fr}}

\address{
LTCI, CNRS, T\'el\'ecom ParisTech,\\
Universit\'e Paris-Saclay,\\
75013, Paris, France\\
\printead{e3}
}

\runauthor{Ch\'etelat, Lederer, Salmon}

\end{aug}

\begin{abstract}
High-dimensional prediction typically comprises two steps: variable selection and subsequent least-squares refitting on the selected variables. However, the standard variable selection procedures, such as the lasso, hinge on tuning parameters that need to be calibrated. Cross-validation, the most popular calibration scheme, is computationally costly and lacks finite sample guarantees. In this paper, we introduce an alternative scheme, easy to implement and both computationally and theoretically efficient.
\end{abstract}


\tableofcontents
\end{frontmatter}

\section{Introduction}\label{sec:intro}
Variable selection has become a basic tool for estimating linear models on large data sets. The most popular method for variable selection is the lasso~\cite{Tibshirani96}, which minimizes the sum of squares errors under an $\ell_1$-penalty. Although efficient at selecting variables when properly tuned, the lasso has the disadvantage that all coefficients are shrunk towards zero. To mitigate this bias, practitioners typically rely on a two-step estimation of the coefficients by computing a least-squares estimate on the variables selected by the lasso.

For illustration, consider the leukemia micro-array data set of~\cite{Golub99}, which consists of $n=38$ bone marrow samples analyzed with $p=7129$ probes from several thousand human genes. A particular interest is to predict the type of leukemia (AML or ALL) present in a patient. The data set also contains an independent test set of 34~observations that are used for assessment of the predictive performance.

In this problem, there are many more variables (7129 features) than available observations (38 samples), and in such a context, a least-squares fitting is not appropriate. A standard solution is to perform variable selection using the lasso, with tuning parameter chosen by cross-validation on the prediction loss. However, since the lasso is known to involve a bias, practitioners commonly refit a least-squares estimate on the selected variables. If the lasso tuning parameter is chosen using 10-fold cross-validation, this approach, called lassoCV in the following, yields a prediction risk of $0.36$ on the test set, computed in 463 seconds.





Although common among practitioners, this approach is suboptimal, because the cross-validation does not take into account the least-squares refitting. Another alternative is to tune the cross-validation for the entire two-step procedure. On the test set, this approach with 10-fold cross-validation, called lslassoCV in the following, yields a prediction risk of $0.45$ computed in 499 seconds.

This adjusted approach is natural, yet suffers from two drawbacks. First, every cross-validation fold must fit a least-squares on each subset selected on the lasso path, which becomes computationally intensive once larger data sets are considered. Second, the method does not come with theoretical guarantees, an issue shared by most cross-validation procedures.

To address these problems, we propose Adaptive Validation for Prediction, (\avt), a novel variable selection scheme. A pseudo-code description of the algorithm is given as Algorithm \ref{sec:intro}. Our proposal is closely related to the recently introduced $\ell_\infty$-Adaptive Validation (\avi)~scheme \cite{Lederer14a}, which is based on tests inspired by isotropic versions of Lepski's method~\cite{Chichignoud_Lederer14, Lepski90, Lepski_Mammen_Spokoiny97}. This approach has been shown to provide fast and optimal calibration of the lasso for (one-step) estimation with respect to $\ell_\infty$-loss. For the two-step prediction considered in this paper, however, a considerably different and more technical approach inspired by non-isotropic tests is required.

As a practical example, Figure~\ref{fig:prediction_timings} compares  \avt\ and  standard methods on the Leukemia dataset discussed above. The methods under consideration to select the lasso tuning parameter are \avt, 10-fold cross-validation, Bayesian Information Criterion (lassoBIC), and an estimator obtained by selecting with BIC a least-square estimator over the supports geneterated by a the lasso path (lslassoBIC) following~\cite{Bellec16} (see the Appendix for further information about the implementation of the latter approach). As can be seen, \avt~is faster (39 seconds) to compute than cross-validation, and it is nearly as fast as the lassoBIC (40 seconds) and lslassoBIC (56 seconds). At the same time, it rivals the predictive performance of all competing approaches. Note at this point that the lassoBIC is a variable selection method rather than a predictive method; two goals that can be considerably different from each other.

The organization of this article is as follows. In the next section, we introduce the algorithm and prove that \avt~predictions satisfy an oracle inequality, that is, are optimal up to a constant factor. In Section~\ref{sec:experiments}, we show that on simulations, \avt\ is substantially faster than cross-validation while being comparable in accuracy.

All proofs are deferred to the Supplementary Material.

\begin{figure}[t]
\centering
\begin{minipage}{0.48\linewidth}
\includegraphics[width=6.1cm]{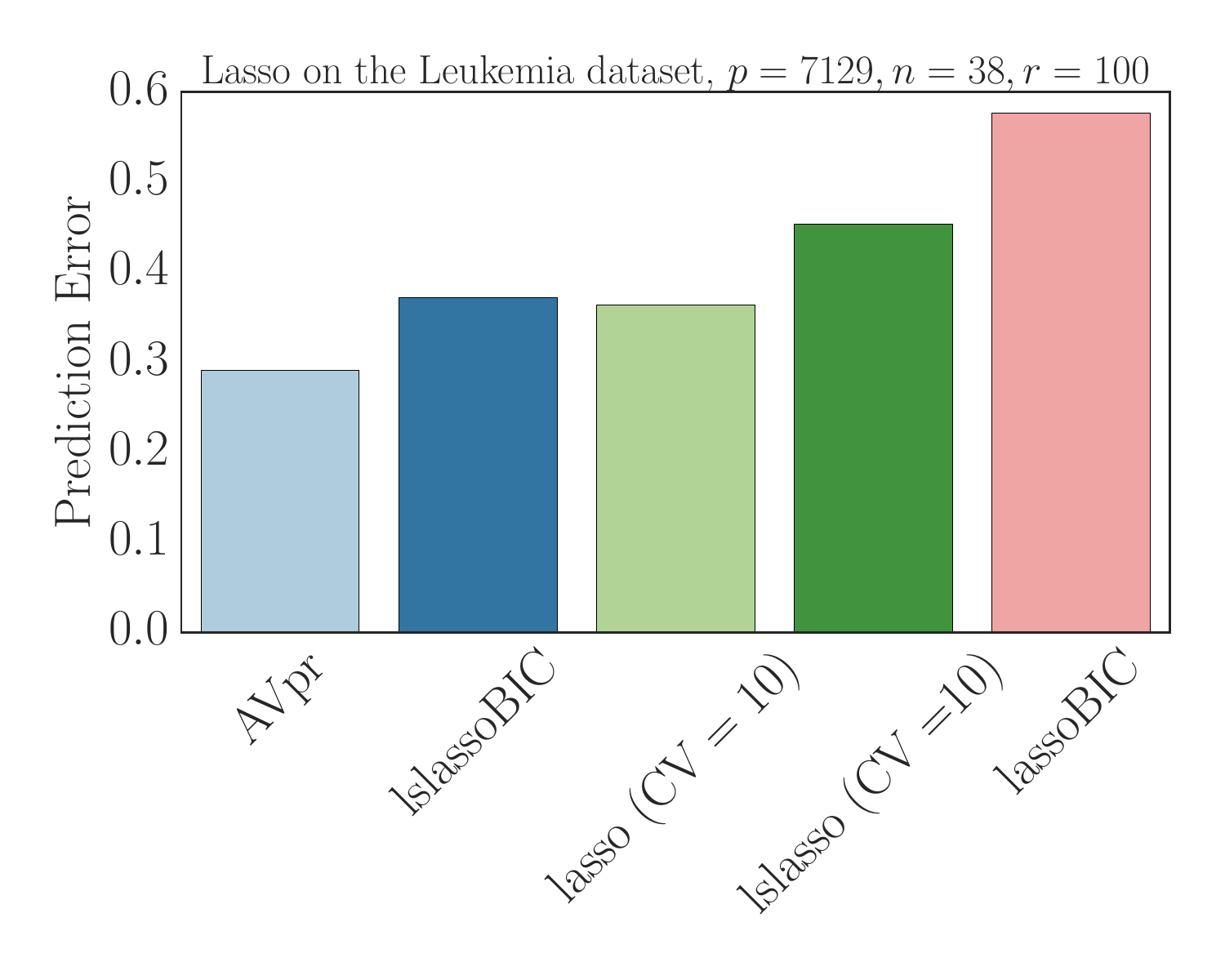}
\end{minipage}
\hspace{0.1cm}
\begin{minipage}{0.48\linewidth}
\includegraphics[width=6.1cm]{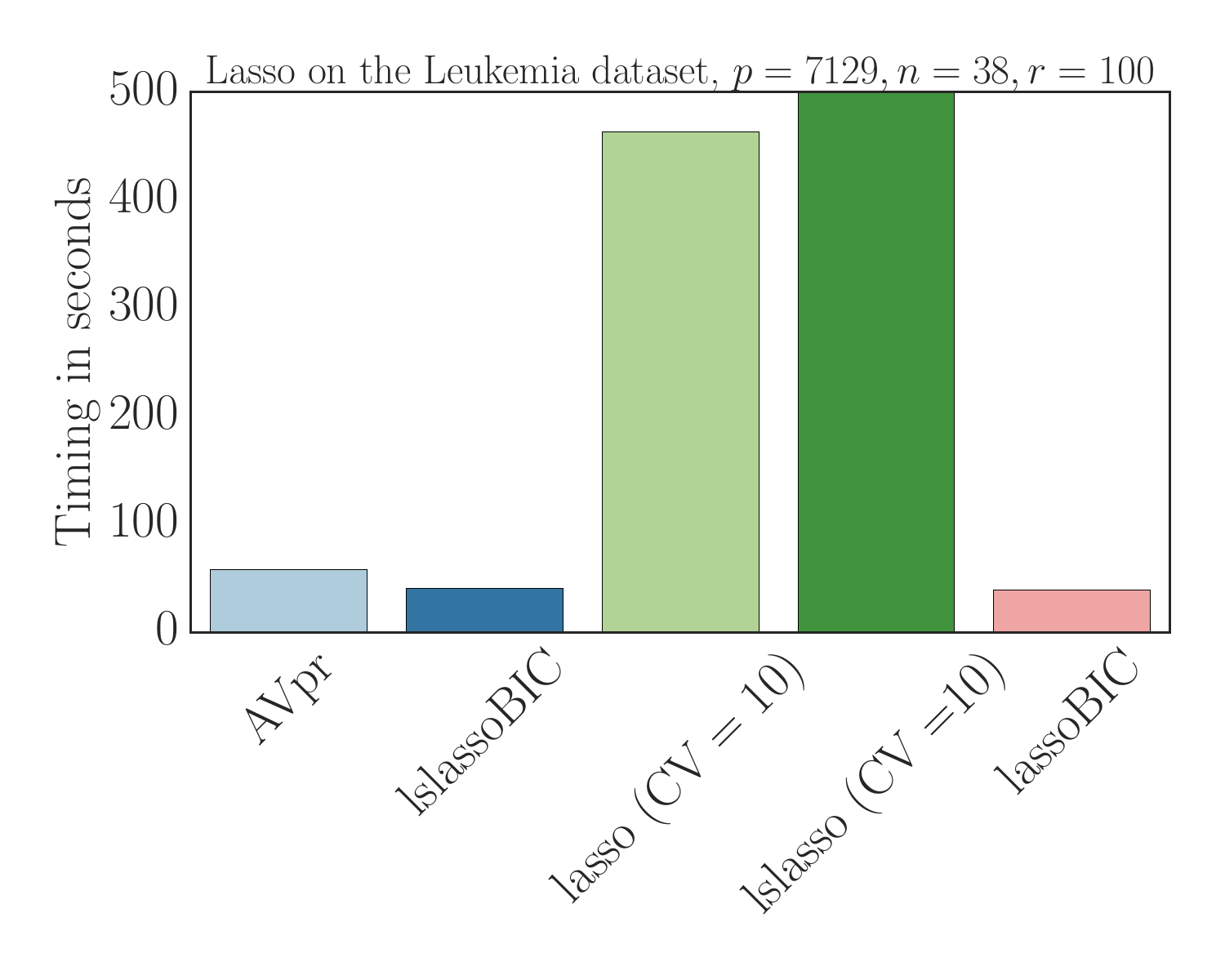}
\end{minipage}
\caption{Prediction error and computing times of the \avt, lslassoBIC \cite{Bellec16}, lassoCV, lslassoCV, and lassoBIC procedures. The bars represent the prediction error on the 34 left out observations. Note that the grid for the tuning parameter contains the same $50$ values for all methods.}
\label{fig:prediction_timings}
 \end{figure}


\subsection*{Framework and Notation}
 Let us describe the framework and the notation. We are interested in linear regression models of the form
 \begin{align}\label{eq:model}
 Y=X\poi+\noise,
 \end{align}
where $Y\in\rn$ is the data, $X\in\R^{n\times p}$ the design matrix, $\poi\in\rp$ the regression vector, and $\noise\in \rn$ the random noise. For ease of exposition, we assume that the noise is Gaussian with unknown variance $\sigma^2$, that is $\noise\sim\mathcal N(0,\sigma^2)$. We assume that the columns of the design matrix $X_1,\dots,X_p\in\rn$ have been standardized to have Euclidean norm $\|X_j\|_2=\sqrt n$, but we otherwise allow for arbitrary correlations between the columns and noise distributions.  We are mainly motivated by (but not limited to) high-dimensional settings with sparse regression vectors, where the number of parameters $p$ can rival or even exceed the number of samples $n$. We finally denote by $S:=\operatorname{supp}[\poi]:=\{j \in [p]:\poi_j\neq 0\}$ the true support, whose cardinality is usually smaller than $n$ and $p$ (where throughout the paper $[d]$ stands for the set $\{1,\dots,d\}$).

A standard approach to find a vector $\hat \poi$ with small prediction loss $\|X\hat \poi -X\poi\|_2^2/n$ is performing a least-squares refitting to the lasso. After reducing the initially large set of variables to a small number of relevant ones, the subsequent refitting aims to lessen the bias associated with the lasso. For a fixed tuning parameter $\lambda$, the lasso~$\hat\poi^\lambda$ is defined via the minimization of objective function
\begin{align}&
\hat\beta^\lambda\in\underset{\theta\in\mathbb{R}^p}{\arg\min}\;\left\{\left\|Y-X\theta\right\|_2^2+2\lambda\|\theta\|_1\right\} \eqsp.
\label{def:lasso}
\end{align}
For simplicity, we will assume that the support of the minimizer equals the equicorrelation set (see Supplementary Material for details). The subsequent least-squares refitting is defined as a minimizer of
\begin{align}&
\bar\beta^\lambda\in\underset{\substack{\supp[\theta]=\,\supp[\hat\beta^\lambda]}}{\arg\min}\;\left\|Y-X\theta\right\|_2^2 \eqsp.
 \label{eq:refitting}
\end{align}
We call this estimator the least-squares lasso (lslasso). This two-step procedure is very popular as it has smaller bias than the lasso for a range of models~\citep{Belloni_Chernozhukov13, Lederer13}.

Our goal is to find optimal tuning parameters for the lslasso~\eqref{eq:refitting} in terms of prediction. In practice, only finitely many estimators can be computed. Therefore,  we consider finite sets of tuning parameters $\Lambda=\{\lambda_1,\dots,\lambda_r\}$, $r\in \mathbb{N}$ and the associated supports $(\hat S^1,\dots,\hat S^r),$ $\hat S^i:=\supp[\hat\beta^{\lambda_i}]$. We denote the collection of supports by $\mathcal S:=\{\hat S^i:~i\in [r]\}$. Finally, we introduce surrogate sets $\hat S^{i,j}:=\hat S^i\cup\hat S^j$ and corresponding estimators
\begin{align}\label{eq:estimator}
\br^{i,j}\in\argmin_{\text{supp}[\xi]\subset\hat S^{i,j}}\|Y-X\xi\|_2^2.
\end{align}
In the special case $i=j$, it holds that $\hat S^{i,j}=\hat S^i$, and hence, $\br^{i}:=\br^{i,i}=\br^{\lambda_i}.$

The lasso is only one out of many variable selection procedures. Our algorithms and derivations can be easily adapted to other procedures, such as the square-root lasso \citep{Antoniadis10, Belloni_Chernozhukov_Wang11, Bunea_Lederer_She13}, scaled-lasso variants \citep{Owen07,Stadler_Buhlmann_vandeGeer10,Sun_Zhang12} or thresholded ridge regression~\citep{Shao_Deng12}, combined with subsequent least-squares refitting. Note for instance that by one-to-one correspondence, our results also hold for the square-root lasso. However, due to its popularity, we focus here only on the lasso.

\subsection*{Related Literature}
Besides the references to the  papers that are most closely connected with our study, we provide some additional pointers to related literature. A discussion of multi-stage methods for regression can be found in~\cite{Wasserman09}. Approaches to tuning parameter calibration in the single-stage setting include~\cite{FloriNeuro11,Chatterjee15,Giraud2012high,Lederer14,Meinshausen_Buhlmann10,Sabourin15,Shah13}. Related papers that appeared recently include \cite{Wang15}, which contains an alternative to least-squares refitting,
and~\cite{Bellec16}, which discusses BIC-type selection as well as Q-aggregation approaches to model selection over the lasso path. The latter contains sparse oracle inequalities as well as prediction and estimation bounds under the standard restricted eigenvalue condition \cite{Bickel_Ritov_Tsybakov09} -  both for (a refitted) BIC-type procedure and for a Q-aggregation procedure. These methods enjoy similar theoretical guarantees as the ones we provide for \avt, and they are also subject to the same issue, namely, that a preliminary estimate of the noise level is required.


\let\relint\undefined
\let\cl\undefined

\newcommand{\Bl}{\mathcal B^\Lambda}
\newcommand{\cl}{\text{cl}\,}
\newcommand{\relint}{\text{relint}\,}
\newcommand{\sgn}{\text{sgn}\,}
\newcommand{\rk}{\text{rk}\,}
\newcommand{\intr}{\text{int}\,}

\newcommand{\oin}{{i^*}}
\newcommand{\ose}{{S^*}}
\newcommand{\obe}{{\beta^*}}

\section{\avt\ and Its Properties}\label{sec:lepski}

\subsection{The \avt\ Algorithm}
The \avt\ scheme is summarized in Algorithm~\ref{alg:avp}. As inputs, it takes the data $(Y, X)$, a set of tuning parameters $\Lambda$, and a constant $a>0$ specified in the following section. It then conducts simple tests along the tuning parameter path of the lasso until a stopping criterion is met. It returns the index of the current tuning parameter~$\overline{i}$ as well as the corresponding two-stage estimator~$\overline{\beta}^{\overline{i}}.$

The algorithm requires the computation of a single lasso path and least-squares estimators along this path. The computation of the paths can be conducted with readily available, easy-to-use, and highly efficient software such as glmnet (in R) or scikit-learn(in Python)~\cite{Friedman10,Pedregosa_etal11}.
For the computation of the least-squares estimators, off-the-shelf solvers can be used since the number of active variables of the second step is typically small.

In Section~\ref{sec:oracleinequality}, \avt\ is shown to satisfy an optimal finite sample prediction bound, and the practical performance of \avt\ is illustrated in Section~\ref{sec:experiments}.

\begin{algorithm}[t]
\caption{\avt}
  \SetAlgoLined 
\vspace{1mm}
 \KwData{$Y, X, \Lambda=\{\lambda_1,\dots,\lambda_r\}, a$}
 \KwResult{$\ein\in [r], \br\in\rp$}
Initialize index: $i \leftarrow 1$\\
Compute $\overline{\beta}^1,\ldots, \overline{\beta}^r$\\
If needed, re-sort the estimators such that $|\hat S^1|\leq\dots\leq|\hat S^r|$\\
\While{$i\leq r-1$\vspace{1mm}}{
Initialize stopping criterion: $TestFailure \leftarrow False$\\
Initialize comparisons: $j \leftarrow i+1$\\
\While{($j \leq r$) \and (TestFailure==False)}{
Compute $\hat{S}^{i,j}$ and $\overline{\beta}^{i,j}$\\
\eIf{$\|X\overline{\beta}^i-X\overline{\beta}^{i,j} \|_2^2\leq a|\hat S^{i}|+a|\hat S^{i,j}|$}{
$j\leftarrow j+1$\\}
{$TestFailure \leftarrow True$}
}
\eIf{$TestFailure==True$}{{$i\leftarrow i+1$}}{break}
}
Set output: $\ein\leftarrow i$ and $\br\leftarrow \br^{\ein}$
\vspace{0.2cm}
\label{alg:avp}
\end{algorithm}

\subsection{Assumption~\ref{ass:app}}

 Let us first introduce and motivate an assumption that ensures a certain stability of lasso solution. In general, if an estimator is unstable for data very close to the (noiseless) underlying truth, accurate estimation and inference hardly seem realistic. For the goal of refitting, we thus introduce an assumption that ensures the stability of supports. In the case of the lasso, this means that we restrict $X\beta$ from being too close to hyperplanes generated by the geometric arrangement of the columns in~$X.$ Figure~\ref{fig:zones} contains a schematic picture of this: $X\beta$ needs to lie outside of small neighborhoods (depicted in orange) around the black boundaries that represent the geometry of $X.$ Most importantly, we stress the assumption does not imply restrictions on the correlations of the design, and does not require estimated  supports to be accurate.

While the assumption concerns the model, it is most convenient to put the precise formulation in terms of the lasso itself. For this, recall that for a fixed~$X$, the support of the lasso evaluated at a vector $z\in\R^n$ is determined by which ``zone'' of $\R^n$ the vector $z$ falls into~\cite{Harris15,Tibshirani_Taylor12}.  These zones exactly correspond to the zones in Figure~\ref{fig:zones} that are separated by the black lines. Importantly, note that we do not require additional variable selection guarantees for the lasso, but merely that the selection is unambiguous. We now define
\begin{align*}
D&:~\R^n\to [0,\infty)\\
D(z)&:=\inf\{\|z-z'\|_\infty/\sqrt n\,: \,z'\in\R^n\text{ s.t. for some }\lambda\in\Lambda,
\\&\qquad \supp[\hat\beta^\lambda(z)]\not=\supp[\hat\beta^{\lambda'}(z')]\text{ for all }\lambda'\in\Lambda\}\ .
\end{align*}
\noindent The function $D$ quantifies how far away a signal can be from the zone boundaries. The factor $1/\sqrt n$ in the definition reflects our normalization of the design matrix. {We also stress that the function involves lasso solutions only at fixed, non-random vectors $z,z'$; in particular, $D$ is independent of~$\varepsilon$ and~$Y$.}

\begin{figure}[htbp]
  \centering
  \includegraphics[width=0.40\textwidth]{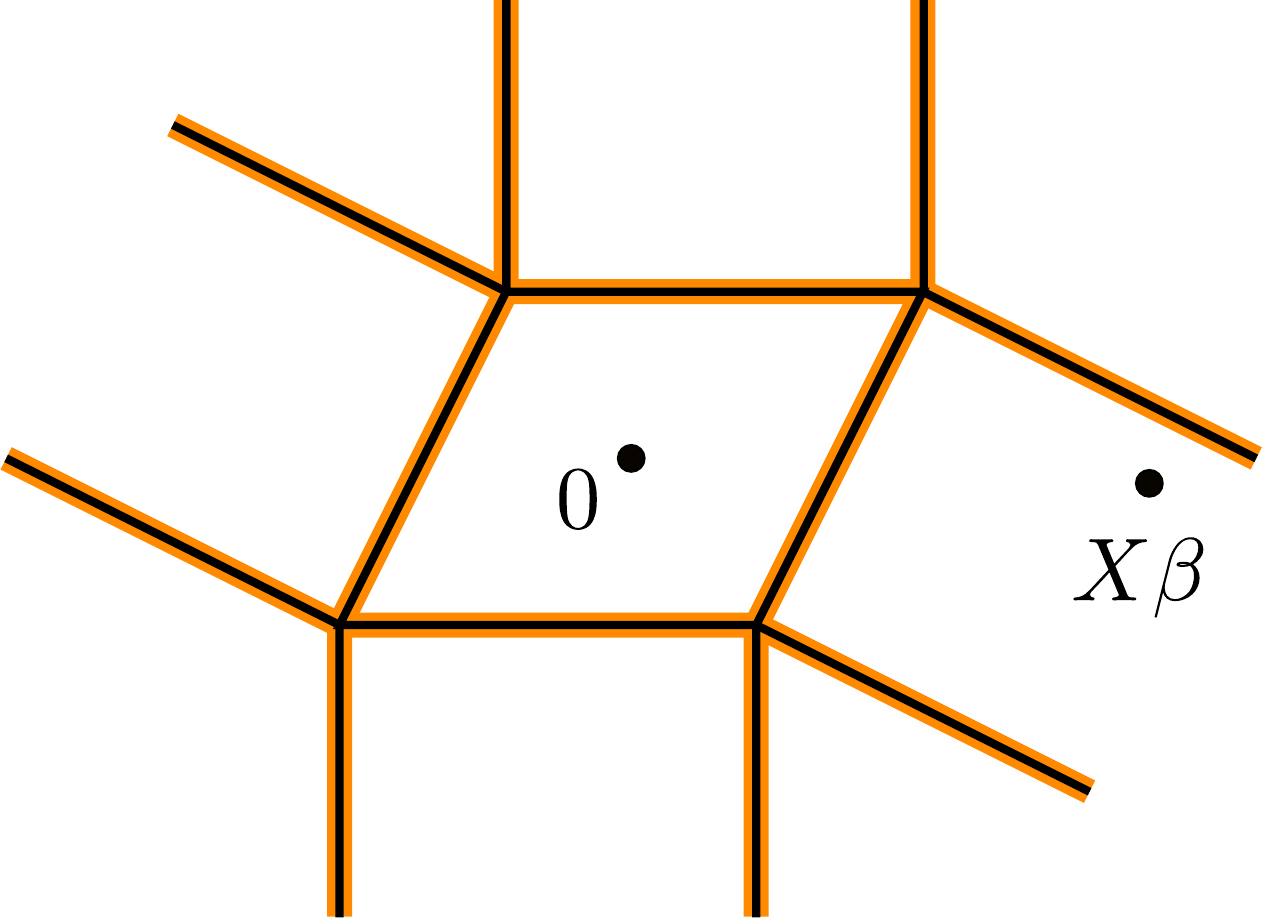}
  \caption{An illustration of Assumption~\ref{ass:app}: $X\beta$ needs to be separated from the boundaries of the zones that determine the active set of the lasso.}
\label{fig:zones}
\end{figure}
\begin{assumption}\label{ass:app} There is an integer $N$ such that for all $n\geq N$, it holds that
  \begin{equation*}
    D(X\beta)> \sqrt{\frac{6\sigma^2\log n}{n}}\ .
  \end{equation*}
\end{assumption}
\noindent This assumption now ensures that $X\beta$ is sufficiently far from the zone boundaries. {Note that the assumption is very different from restricted eigenvalues~\cite{Bunea_Tsybakov_Wegkamp07b} or similar hypothesis in the theory for the lasso~\cite{Dalalyan_Hebiri_Lederer17,vandeGeer_Buhlmann09}. While the latter assumptions need to be strict to ensure a good performance of the lasso, our assumption only requires that the estimates are unambiguous.} In the specific case where $X=\operatorname{I}_{n\times n}$, some insight can be obtained, since the quantity $D(\beta)$ can be computed. Indeed, with the convention that $|\beta_{(s)}|$ is the $s$-th largest amplitude of the vector $|\beta|$, $D(\beta)$ represents the smallest difference $|\beta_{(s)}| -|\beta_{(s+1)}|$, where $s$ is a support size of a Lasso solution applied on $\beta$ (\ie a soft-thresholded version of $\beta$) for a threshold $\lambda \in \Lambda$. In this case, the assumption represents bounding by below such differences, which makes it clear to be an assumption on the underlying signal itself.

To motivate this assumption further, we finally show that a slightly weaker version of Assumption~\ref{ass:app} automatically holds for all $X\beta$ up to a set of measure zero.
\begin{theorem}\label{thm:lasso-dpositive} For all $X\beta\in\R^n$ up to a set of Lebesgue measure zero, the lasso satisfies
  \begin{equation*}
 D(X\beta)>0\ .
  \end{equation*}
\end{theorem}
\noindent Theorem~\ref{thm:lasso-dpositive} does not completely exclude cases that violate Assumption~\ref{ass:app}. However, together with the above discussion, the result indicates that these cases are hardly generic and of limited relevance in applications. 

\subsection{Oracle Inequality}\label{sec:oracleinequality}
Oracle inequalities are bounds for the risk of an estimator. More precisely, they compare the risk of an estimator with the risk of an oracle estimator, an estimator that has knowledge of the best model~\cite{Buhlmann_vandeGeer11,Koltchinskii11}.

In this section, we show that our estimator~\avt\ satisfies such an oracle inequality. To this end, we first introduce the oracle set.

\begin{definition}[Oracle] The oracle set $\ose \in \mathcal S$ is the set $\ose :=\hat S^\oin$ with index
\begin{equation*}
\oin:=\min\big\{{i\in [r]} : \hat S^i\supset S\big\}\ .
\end{equation*}
and the associated oracle estimator is $\beta^*:=\overline{\beta}^{i^*}$.
\end{definition}
\noindent In other words, the oracle set contains the true support $S$ and has minimal cardinality among all such sets. The oracle set can therefore be viewed as the best possible approximation of $S$ in $\mathcal S$; in particular, $S^*=S$ whenever $S\in\mathcal{S}$.

{We implicitly assume that the oracle set exists, that is, the true support set is a subset of an estimated support along the path. However, one can easily generalize the definition to avoid this assumption. Let $S^*$ be an arbitrary set  and $P^*$ the projection onto the space spanned by the columns with indexes in $S^*.$ Adding this projection in our proofs (cf.~\eqref{eq:start-bd} for example) yields the same results as below except for an additional term $\|(\operatorname{I}-P^*)\beta\|^2_2$ in the bounds. However, as the above definition exists in generic cases (since the lasso supports tend to be very exhaustive for small tuning parameters), and as it provides a concise formulation of the results, we do not consider  the extended version in the following.}

Now, given the oracle, we can state a bound for the two-step lasso procedure with the \emph{optimal} tuning parameter, that is, the tuning parameter that leads to the oracle set. Throughout this section we invoke Assumption~\ref{ass:app}, {which helps us rule out ambiguous design settings.}

\begin{proposition}\label{prop:oracle-benchmarkshort} Under Assumption \ref{ass:app}, for any $\alpha>0$, there exist constants $t,N,R>0$ such that for all $n\geq N$ and $r\geq R,$ the oracle estimator satisfies with probability at least $1-\alpha$ the bound
\begin{align*}
\frac{\|X\obe-X\poi\|_2^2}n\;\leq\;&(1+t\log r)\frac{\sigma^2|\ose|}n\ .
\end{align*}
\end{proposition}
\noindent This is a bound for the lasso with refitting - under the assumption that the oracle set $S^*$ is known and incorporated in the selection of the tuning parameter. The constants $t,N,$ and $R$ are specified in the proof section.

In practice, we do not have access to the oracle set $S^*$. Therefore, we hope to find a procedure that does not require its knowledge and still satisfies the bound (up to constants) stated in Proposition~\ref{prop:oracle-benchmarkshort}. The following result shows that \avt\ provides this.

\begin{theorem}[Oracle inequality for \avt]\label{thm:oracle-inequalityshort}
If Assumption~\ref{ass:app} holds, for any $\alpha>0$, there exist constants $t,N,R>0$ such that for all $n\geq N$ and all $r>R$, our estimator \avt\ with $a\geq2\sigma^2(1+t\log r)$ satisfies with probability at least $1-\alpha$ the bounds

  \begin{align*}
    \tag{i}&~~~~~~~|\hat S |\leq |\ose|\\
  \text{and }\quad\tag{ii}&\frac{\|X\br-X\poi\|_2^2}n\;\leq\;\Big[6a+4\sigma^2(1+t\log r)\Big]\frac{|\ose|}n\ .
  \end{align*}
\end{theorem}
\noindent {This proves optimality of \avt: indeed, if $a\gtrsim2\sigma^2(1+t\log r),$ \avt\ satisfies the same bound (up to constants) as the two-step approach that is based on the knowledge of the oracle set~$S^*.$ Explicit constants can be found in the proofs section, though we did not attempt to optimize them.}

Theorem~\ref{thm:oracle-inequalityshort} holds for any sufficiently large $a.$ The question is now how to choose $a$ in practice. Theorem~\ref{thm:oracle-inequalityshort} entails precise guidance for this choice. In view of the bounds, one should select the smallest~$a$ that is still allowed, that is, $a=2\sigma^2(1+t\log r).$ However, since $\sigma^2$ is typically unknown in practice, we suggest to replace it with a rough estimate~$\hat\sigma$. Moreover, we argue that the term \jrs{$2(1+t\log r)$} is an artifact of our proof technique rather than a fundamental aspect of the bound. We thus suggest the simple choice \jrs{$a=\hat\sigma^2$}, see the empirical section below. Consequently, the bounds above provide a solid theoretical foundation for~\avt; however, there is still a gap between theory and practice that deserves to be studied further.

We note that our approach is very different from just replacing the unknown noise variance in the existing theoretical tuning parameters. Standard oracle inequalities for the lasso hold true with probability $t$ for tuning parameters of the form $\text{const}_t\,\sigma\sqrt{(\log p)/n}$, where $\text{const}_t$ is a factor involving the level $t$, see~\cite{Buhlmann_vandeGeer11} and references herein. Thus, one might be tempted to use these tuning parameters with an estimate of~$\sigma$. However, the above form is valid only for Gaussian noise, while we aim at more general calibration.  Moreover, even for Gaussian noise, the above form is known to be suboptimal both in the near orthogonal case (because $p$ could be replaced by $p/s,$ where $s$ is the true sparsity level) and in the correlated case (where much smaller tuning parameters might be favored), we refer to \cite{Buhlmann_vandeGeer11,Dalalyan_Hebiri_Lederer17,Hebiri_Lederer13} and references therein. Finally, even if the above form were optimal in terms of the standard oracle inequalities for prediction, estimation, and variable selection, there are no guarantees on their performance in terms of refitting.


\newcommand{\told}[1]{\textcolor{green}{#1}}
\newcommand{\tnew}[1]{\textcolor{red}{#1}}
\newcommand{\rem}[1]{\textcolor{red}{\xout{#1}}}
\newcommand{\tsupernew}[1]{\textcolor{blue}{#1}}

\section{Experiments}\label{sec:experiments}

\subsection{General Setup}

We measure the numerical performance and the computational speed of \avt\ in two-step prediction. The methods of comparison are cross-validation with 2, 5, 10, and 20 number of folds, which are typically regarded as the standard calibration schemes. 

Variable selection is performed with the lasso. We emphasize that the motivation of this work is not to compare different variable selection methods, but instead, to compare different calibration schemes in two-step prediction.


\newcommand{\sigtol}{\ensuremath{\delta}}
\newcommand{\sigest}{\ensuremath{\widehat{\sigma}}}

\begin{algorithm}[t]
\DontPrintSemicolon
  \SetAlgoLined 
\vspace{1mm}
 \KwData{$Y, X, \sigtol$}
 \KwResult{$\sigest$}
Initialize tuning parameter and variance: $\lambda_0\leftarrow\sqrt{2n\log p}$ and $\sigest\leftarrow1$\\
\Repeat{$|\sigest-\sigest'|\leq \sigtol$\vspace{1mm}}{
Save $\sigest'\leftarrow \sigest$\\
Update $\sigest$:\\
~~~~Set $\lambda \leftarrow \sigest\lambda_0$\\
~~~~Compute $\hat\beta^\lambda$ as the lasso with tuning parameter $\lambda$ according to \eqref{def:lasso}\\
~~~~Set $\sigest \leftarrow \|Y-X\hat\beta^\lambda\|_2/\sqrt{n}$\\}
\caption{Scaled lasso algorithm with early stopping, cf. \cite{Sun_Zhang12}}
\label{algosqrt}
\end{algorithm}

The data are generated according to a linear regression model as in~\eqref{eq:model} with  $n=p=100,200$.
The first~$10$ entries of the regression vector~$\beta$ are set to~$1$, while all other entries are set to~$0$.
The components of the noise vector are independently sampled from a univariate standard normal distribution with mean~$0$ and variance~$1$.
The rows of the design matrix~$X$ are independently sampled from a multivariate normal distribution with mean~$0$
and covariance matrix~$\Sigma$ that is set to $\Sigma_{ij}=1$ for $i = j$ and to $\Sigma_{ij}=\rho$ for $i\neq j$ with~$\rho=0.5$. Subsequently, the columns of~$X$ are normalized to Euclidean norm~$\sqrt n$. For all experiments, we perform~$50$ repetitions.

In addition to the described parameter settings, we tested various other settings, including different
correlation coefficients~$\rho$, regression vectors~$\beta$, and tuning parameter grids.
As the conclusions were similar across all settings, we restrict our presentation to the ones described. All computations are conducted with the standard implementations of the lasso from Python \textit{scikit-learn} (version 0.16)~\cite{Pedregosa_etal11}, and our code is available at \texttt{https://github.com/josephsalmon/AVp}.

\subsection{Practical Choice of $a$} {We follow the suggestions after Theorem~\ref{thm:oracle-inequalityshort}. Specifically, if} $\sigma^2$ is known,
we recommend using Algorithm~\ref{alg:avp} with $a=\sigma^2$
as suggested by Theorem~\ref{thm:oracle-inequalityshort}
(regarding the term $t\log r$ as a superfluous term coming from our proof technique). In practice, however, the noise variance~$\sigma^2$ is often unknown. We then advocate using $a=\sigest^2$ with a rough estimate~$\sigest^2$ of~$\sigma^2.$ Such a rough estimate can be easily obtained by using a (very) small number of iterations of the algorithms for the square-root lasso~\cite{Bunea_Lederer_She13} or
the scaled lasso~\cite{Sun_Zhang13}. For our simulations, we have opted for the latter, which consists of an alternating minimization for estimating
both the regression parameter and the noise level. Algorithm~\ref{algosqrt} states our concrete implementation. We set the tolerance to $\sigtol=10^{-2}$,  which typically leads to less than five iterations of the loop in the algorithm and therefore, as illustrated below,  to very low computational costs.

\subsection{Choice of the Tuning Parameter Grids}\label{subsec:grid} The tuning parameter grid is chosen as a default grid in the lasso function in \textit{scikit-learn}. More precisely, we take a geometric grid of size $r=100$  starting from $\lambda_{\max}=\|X^\top Y\|_\infty$,  the smallest tuning parameter that leads to a lasso solution of all zeros, and ending at $\lambda_{\max}/1000$.



\subsection{Computational and Statistical Performance}

We first report in Figure~\ref{fig:timings}  the computational times for each of the following:\\

\begin{tabular}{l l}
lasso path:&
\begin{minipage}[t]{0.6\columnwidth}
Computation of one tuning parameter path of lasso.\vspace{1.5mm}
\end{minipage}\\
lslasso path:& \begin{minipage}[t]{0.6\columnwidth}
Computation of one tuning parameter path of lasso with least-squares refitting.\vspace{1.5mm}
\end{minipage}\\
lasso\avt:&
\begin{minipage}[t]{0.6\columnwidth}
Least-squares refitted lasso with tuning parameter selected by \avt~with $a=\sigest^2$ as detailed above.\vspace{1.5mm}
\end{minipage}\\
lslassoBIC: &
\begin{minipage}[t]{0.6\columnwidth}
Least-squares refitted lasso with tuning parameter selected by a BIC-type procedure \cite{Bellec16}, detailed in Appendix~\ref{sec:lassobic}.\vspace{1.5mm}
\end{minipage}\\
lassoCV:&
\begin{minipage}[t]{0.6\columnwidth}
Least-squares refitted lasso with tuning parameter selected by cross-validation on the estimates of the (one-step) lasso.\vspace{1.5mm}
\end{minipage}\\
lslassoCV:&
\begin{minipage}[t]{0.6\columnwidth}
Least-squares refitted lasso with tuning parameter selected by cross-validation on the estimates of least-squares refitted lasso.\vspace{1.5mm}
\end{minipage}
\end{tabular}\\
We then also report in Figure~\ref{fig:box_plots} the prediction performances of the last three methods.


\begin{figure}[t]\centering
\includegraphics[width=0.95\linewidth]{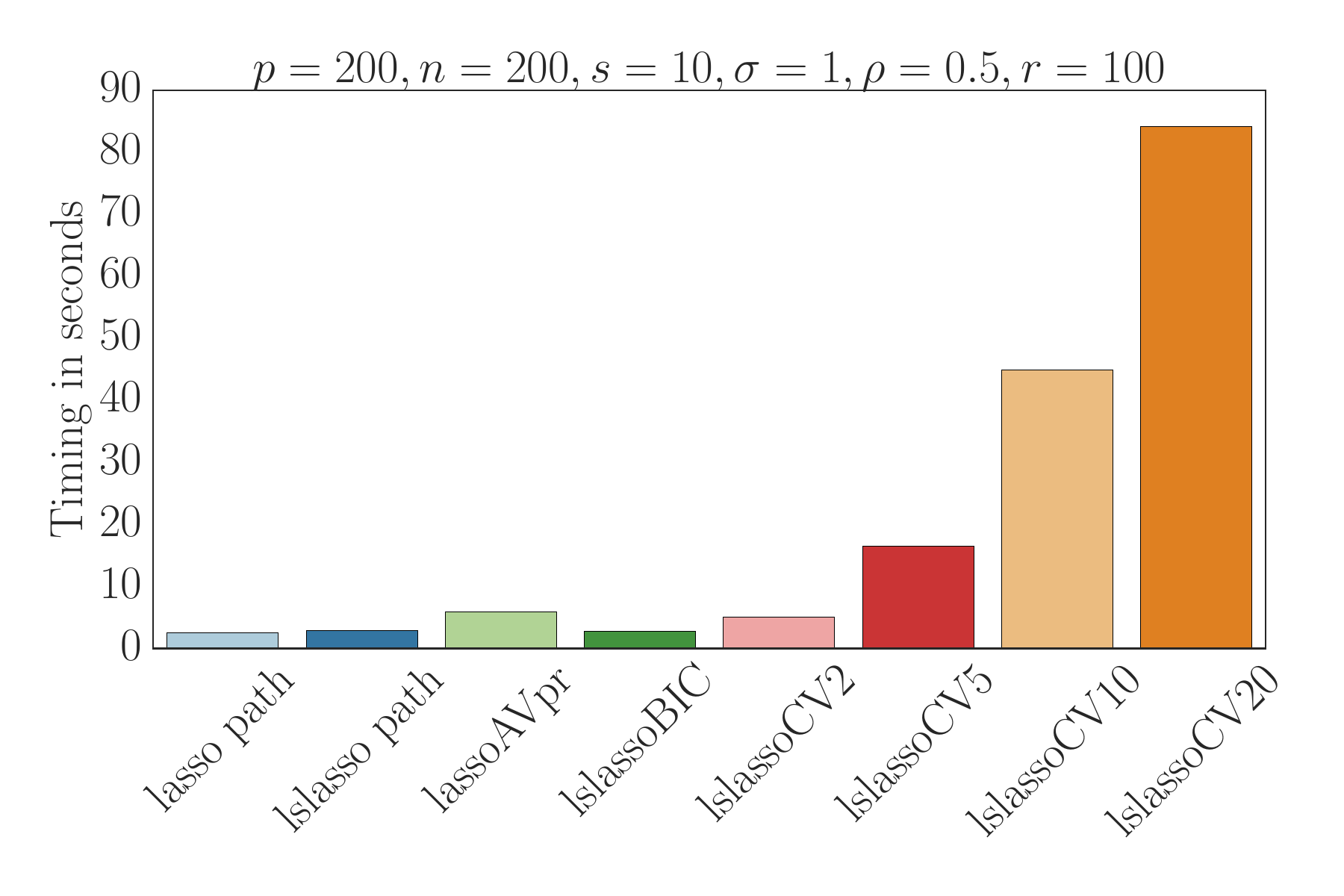}
\caption{Computation times of the lasso path, lslasso path, lslassoCV, and lasso\avt~(with $\sigest$). Cross-validation is performed using a refitting step (lslassoCV) for different numbers of folds.}
\label{fig:timings}
 \end{figure}

In conclusion, our simulations demonstrate that \avt~is competitive both in computational speed and in prediction performance.

\begin{figure}[h]\centering
\includegraphics[width=11.7cm]{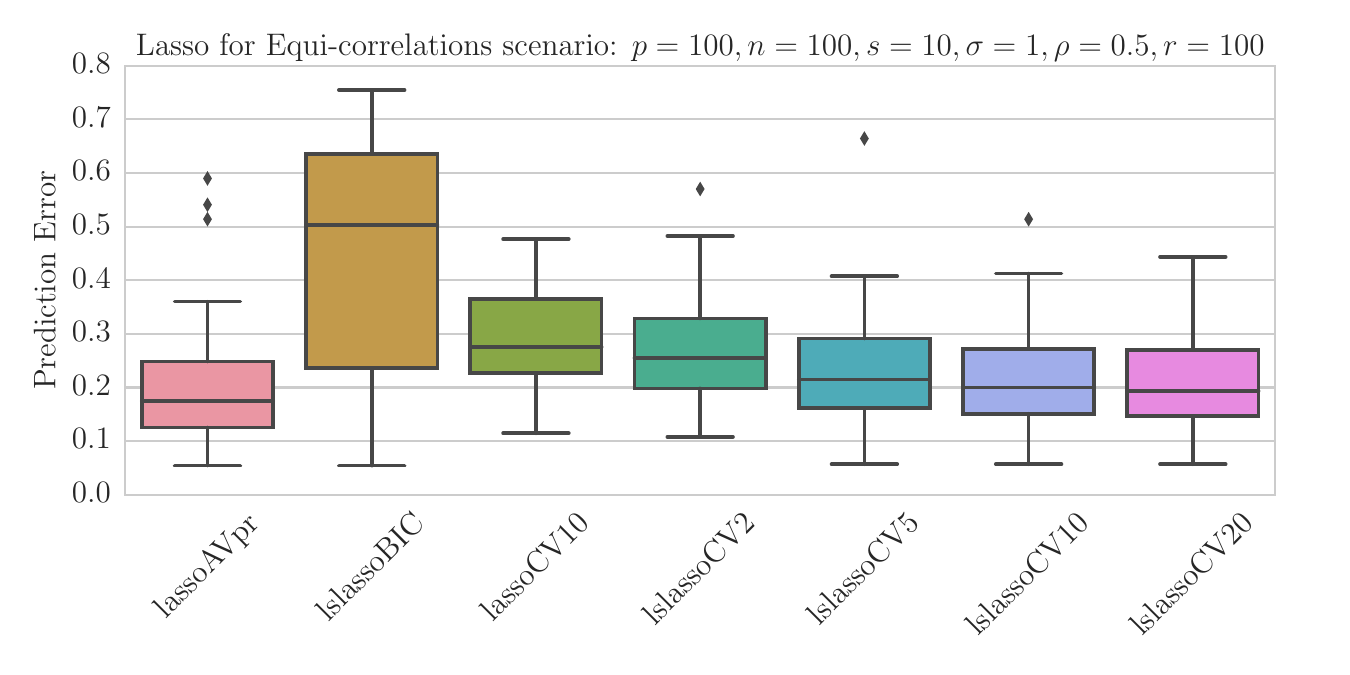}%
\hspace{0.1cm}
\includegraphics[width=11.7cm]{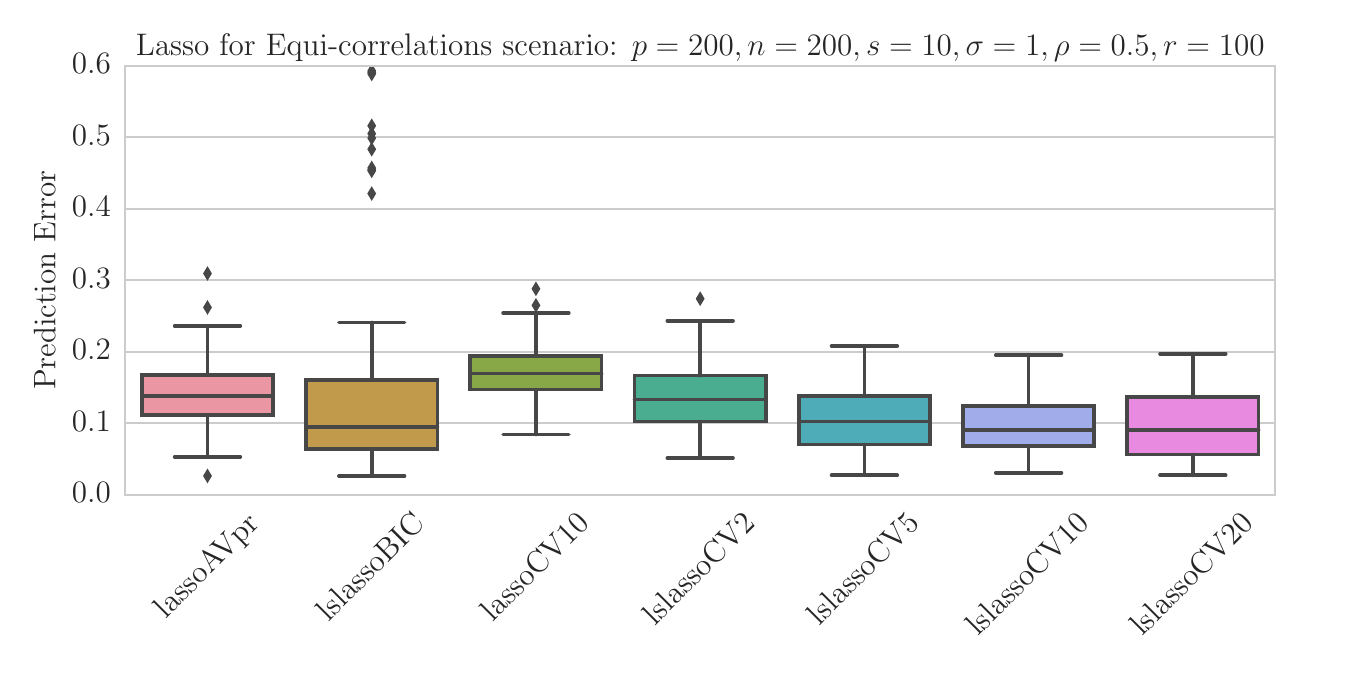}%
\caption{Prediction errors of the lassoCV, lslassoCV, and lasso\avt~(with $\sigest$). Cross-validation is performed using (lslassoCV) or not using (lassoCV) a refitting step and using different numbers of folds.}
\label{fig:box_plots}
\end{figure}


\section{Discussion}\label{sec:conclusions}
The standard scheme for calibrating the lasso is cross-validation. However, cross-validation entails two main deficiencies: it is computationally inefficient and lacks finite sample guarantees. In contrast, \avt\ is fast and satisfies optimal bounds for prediction with the refitted lasso. We therefore propose \avt\ as an alternative scheme for prediction with lasso followed by refitting. {Moreover, our work can be readily extended to the square-root lasso and to ridge regression with thresholding.}

{A direction for further research could be sharpening the theoretical bounds. The current result contains a term that grows logarithmically in the number of tuning parameters under consideration. Such artifacts are common in the non-parametric literature.  In a sense, one might consider our approach as a non-parametric version of~\cite{Chichignoud_Lederer14}.  Thus, it is not surprising that such a term appears.  Empirically, however, there are no indications that this term is needed. An improved understanding of the bound could especially strengthen the connections between the theoretical and the practical choice of $a$.}

The exact specification of the method is an issue that appears more generally in  tuning parameter calibration. In our case, there is flexibility in how to estimate the noise variance; in cross-validation, one has to specify the number of folds; in BIC-type approaches, the constant in front of the $\log$ penalty needs to be adjusted; in Q-aggregation, there is a trade-off between the KL regularization and the quadratic term. We believe that although out of the scope of this contribution, a comprehensive sensitivity analysis studying the selections in each of the methods would be of interest.

\section*{Acknowledgements}
{We thank Micha\"el Chichignoud for the insightful remarks and for the inspiring discussions, and we thank Pierre Bellec for providing us with valuable insights about numerical aspects of his work. We also thank the editors and reviewers for the valuable comments that have improved the paper.}

\bibliographystyle{plainnat}
\bibliography{references_all}

\providecommand{\CH}{C-H} \providecommand{\JM}{J-M}
\begin{thebibliography}{42}
\providecommand{\natexlab}[1]{#1}
\providecommand{\url}[1]{\texttt{#1}}
\expandafter\ifx\csname urlstyle\endcsname\relax
  \providecommand{\doi}[1]{doi: #1}\else
  \providecommand{\doi}{doi: \begingroup \urlstyle{rm}\Url}\fi

\bibitem[Antoniadis(2010)]{Antoniadis10}
A.~Antoniadis.
\newblock Comments on: {$\ell_1$}-penalization for mixture regression models.
\newblock \emph{TEST}, 19\penalty0 (2):\penalty0 257--258, 2010.

\bibitem[Bellec(2016)]{Bellec16}
P.~Bellec.
\newblock Aggregation of supports along the {Lasso} path.
\newblock In \emph{COLT}, pages 488--529, 2016.

\bibitem[Belloni and Chernozhukov(2013)]{Belloni_Chernozhukov13}
A.~Belloni and V.~Chernozhukov.
\newblock Least squares after model selection in high-dimensional sparse
  models.
\newblock \emph{Bernoulli}, 19\penalty0 (2):\penalty0 521--547, 2013.

\bibitem[Belloni et~al.(2011)Belloni, Chernozhukov, and
  Wang]{Belloni_Chernozhukov_Wang11}
A.~Belloni, V.~Chernozhukov, and L.~Wang.
\newblock Square-root {Lasso}: Pivotal recovery of sparse signals via conic
  programming.
\newblock \emph{Biometrika}, 98\penalty0 (4):\penalty0 791--806, 2011.

\bibitem[Bickel et~al.(2009)Bickel, Ritov, and
  Tsybakov]{Bickel_Ritov_Tsybakov09}
P.~Bickel, Y.~Ritov, and A.~Tsybakov.
\newblock Simultaneous analysis of {Lasso} and {D}antzig selector.
\newblock \emph{Ann. Statist.}, 37\penalty0 (4):\penalty0 1705--1732, 2009.

\bibitem[B{\"u}hlmann and {van de Geer}(2011)]{Buhlmann_vandeGeer11}
P.~B{\"u}hlmann and S.~{van de Geer}.
\newblock \emph{Statistics for high-dimensional data}.
\newblock Springer Series in Statistics. Springer, Heidelberg, 2011.
\newblock Methods, theory and applications.

\bibitem[Bunea et~al.(2007)Bunea, Tsybakov, and
  Wegkamp]{Bunea_Tsybakov_Wegkamp07b}
F.~Bunea, A.~Tsybakov, and M.~Wegkamp.
\newblock Sparsity oracle inequalities for the {L}asso.
\newblock \emph{Electron. J. Stat.}, 1:\penalty0 169--194 (electronic), 2007.

\bibitem[Bunea et~al.(2011)Bunea, She, Ombao, Gongvatana, Devlin, and
  Cohen]{FloriNeuro11}
F.~Bunea, Y.~She, H.~Ombao, A.~Gongvatana, K.~Devlin, and R.~Cohen.
\newblock Penalized least squares regression methods and applications to
  neuroimaging.
\newblock \emph{Neuroimage}, 55, 2011.

\bibitem[Bunea et~al.(2014)Bunea, Lederer, and She]{Bunea_Lederer_She13}
F.~Bunea, J.~Lederer, and Y.~She.
\newblock The group square-root {Lasso}: Theoretical properties and fast
  algorithms.
\newblock \emph{{IEEE} Trans. Inf. Theory}, 60\penalty0 (2):\penalty0
  1313--1325, 2014.

\bibitem[Chatterjee and Jafarov(2015)]{Chatterjee15}
S.~Chatterjee and J.~Jafarov.
\newblock Prediction error of cross-validated lasso.
\newblock \emph{arXiv:1502.06291}, 2015.

\bibitem[Chichignoud and Lederer(2014)]{Chichignoud_Lederer14}
M.~Chichignoud and J.~Lederer.
\newblock A robust, adaptive {M}-estimator for pointwise estimation in
  heteroscedastic regression.
\newblock \emph{Bernoulli}, 20\penalty0 (3):\penalty0 1560--1599, 2014.

\bibitem[Chichignoud et~al.(2016)Chichignoud, Lederer, and
  Wainwright]{Lederer14a}
M.~Chichignoud, J.~Lederer, and M.~Wainwright.
\newblock Tuning {L}asso for sup-norm optimality.
\newblock \emph{J. Mach. Learn. Res.}, 17, 2016.

\bibitem[Dalalyan et~al.(2017)Dalalyan, Hebiri, and
  Lederer]{Dalalyan_Hebiri_Lederer17}
A.~Dalalyan, M.~Hebiri, and J.~Lederer.
\newblock On the prediction performance of the {Lasso}.
\newblock \emph{Bernoulli}, 23\penalty0 (1):\penalty0 552--581, 2017.

\bibitem[Friedman et~al.(2010)Friedman, Hastie, and Tibshirani]{Friedman10}
J.~Friedman, T.~Hastie, and R.~Tibshirani.
\newblock Regularization paths for generalized linear models via coordinate
  descent.
\newblock \emph{J. Stat. Softw.}, 33\penalty0 (1):\penalty0 1--22, 2010.

\bibitem[Giraud et~al.(2012)Giraud, Huet, and Verzelen]{Giraud2012high}
C.~Giraud, S.~Huet, and N.~Verzelen.
\newblock High-dimensional regression with unknown variance.
\newblock \emph{Statist. Sci.}, 27\penalty0 (4):\penalty0 500--518, 2012.

\bibitem[Golub et~al.(1999)Golub, Slonim, Tamayo, Huard, Gaasenbeek, Mesirov,
  Coller, Loh, Downing, Caligiuri, Bloomfield, and Lander]{Golub99}
T.~Golub, D.~Slonim, P.~Tamayo, C.~Huard, M.~Gaasenbeek, J.~Mesirov, H.~Coller,
  M.~Loh, J.~Downing, M.~Caligiuri, C.~Bloomfield, and E.~Lander.
\newblock Molecular classification of cancer: class discovery and class
  prediction by gene expression monitoring.
\newblock \emph{Science}, 286\penalty0 (5439):\penalty0 531--537, 1999.

\bibitem[{Gr\"unbaum}(2003)]{Grunbaum03}
B.~{Gr\"unbaum}.
\newblock \emph{Convex Polytopes}.
\newblock Springer-Verlag, New York, second edition, 2003.

\bibitem[Harris and Sepehri(2015)]{Harris15}
N.~Harris and A.~Sepehri.
\newblock The accessible lasso models.
\newblock \emph{arXiv:1501.02559}, 2015.

\bibitem[Hebiri and Lederer(2013)]{Hebiri_Lederer13}
M.~Hebiri and J.~Lederer.
\newblock How correlations influence {Lasso} prediction.
\newblock \emph{IEEE Transactions on Information Theory}, 59:\penalty0
  1846--1854, 2013.

\bibitem[Koltchinskii(2011)]{Koltchinskii11}
V.~Koltchinskii.
\newblock \emph{Oracle inequalities in empirical risk minimization and sparse
  recovery problems}, volume 2033 of \emph{Lecture Notes in Mathematics}.
\newblock Springer, Heidelberg, 2011.

\bibitem[Lederer(2013)]{Lederer13}
J.~Lederer.
\newblock Trust, but verify: benefits and pitfalls of least-squares refitting
  in high dimensions.
\newblock \emph{arXiv:1306.0113 [stat.ME]}, 2013.

\bibitem[Lederer and M{\"u}ller(2015)]{Lederer14}
J.~Lederer and C.~M{\"u}ller.
\newblock Don't fall for tuning parameters: Tuning-free variable selection in
  high dimensions with the trex.
\newblock In \emph{Proceedings of the Twenty-Ninth AAAI Conference on
  Artificial Intelligence}, 2015.

\bibitem[Lee and Sun(2015)]{Lee15}
D.~Lee, J.and~Sun and Y.~Sun.
\newblock Exact post-selection inference, with applications to the lasso.
\newblock \emph{Preprint arXiv:1311.6238v5}, 2015.

\bibitem[Lepski(1990)]{Lepski90}
O.~Lepski.
\newblock On a problem of adaptive estimation in gaussian white noise.
\newblock \emph{Theory Probab. Appl.}, 35\penalty0 (3):\penalty0 454--466,
  1990.

\bibitem[Lepski et~al.(1997)Lepski, Mammen, and
  Spokoiny]{Lepski_Mammen_Spokoiny97}
O.~Lepski, E.~Mammen, and V.~Spokoiny.
\newblock Optimal spatial adaptation to inhomogeneous smoothness: an approach
  based on kernel estimates with variable bandwidth selectors.
\newblock \emph{Ann. Statist.}, 25\penalty0 (3):\penalty0 929--947, 1997.

\bibitem[Mairal and Yu(2012)]{Mairal_Yu12}
J.~Mairal and B.~Yu.
\newblock Complexity analysis of the lasso regularization path.
\newblock \emph{Proceedings of the 29th International Conference on Machine
  Learning}, 2012.

\bibitem[Meinshausen and B{\"u}hlmann(2010)]{Meinshausen_Buhlmann10}
N.~Meinshausen and P.~B{\"u}hlmann.
\newblock Stability selection.
\newblock \emph{J. Roy. Statist. Soc. Ser. B}, 72\penalty0 (4):\penalty0
  417--473, 2010.

\bibitem[Owen(2007)]{Owen07}
A.~Owen.
\newblock A robust hybrid of lasso and ridge regression.
\newblock \emph{Contemporary Mathematics}, 443:\penalty0 59--72, 2007.

\bibitem[Pedregosa et~al.(2011)Pedregosa, Varoquaux, Gramfort, Michel, Thirion,
  Grisel, Blondel, Prettenhofer, Weiss, Dubourg, Vanderplas, Passos,
  Cournapeau, Brucher, Perrot, and Duchesnay]{Pedregosa_etal11}
F.~Pedregosa, G.~Varoquaux, A.~Gramfort, V.~Michel, B.~Thirion, O.~Grisel,
  M.~Blondel, P.~Prettenhofer, R.~Weiss, V.~Dubourg, J.~Vanderplas, A.~Passos,
  D.~Cournapeau, M.~Brucher, M.~Perrot, and E.~Duchesnay.
\newblock Scikit-learn: Machine learning in {P}ython.
\newblock \emph{J. Mach. Learn. Res.}, 12:\penalty0 2825--2830, 2011.

\bibitem[Sabourin et~al.(2015)Sabourin, Valdar, and Nobel]{Sabourin15}
J.~Sabourin, W.~Valdar, and A.~Nobel.
\newblock A permutation approach for selecting the penalty parameter in
  penalized model selection.
\newblock \emph{Biometrics}, 71:\penalty0 1185--1194, 2015.

\bibitem[Schneider(2013)]{Schneider13}
R.~Schneider.
\newblock \emph{Convex bodies: the Brunn--Minkowski theory}, volume 151 of
  \emph{Encyclopedia of Mathematics and its Applications}.
\newblock Cambridge University Press, second edition, 2013.

\bibitem[Shah and Samworth(2013)]{Shah13}
R.~Shah and R.~Samworth.
\newblock Variable selection with error control: another look at stability
  selection.
\newblock \emph{J. Roy. Statist. Soc. Ser. B}, 75\penalty0 (1):\penalty0
  55--80, 2013.

\bibitem[Shao and Deng(2012)]{Shao_Deng12}
J.~Shao and X.~Deng.
\newblock Estimation in high-dimensional linear models with deterministic
  design matrices.
\newblock \emph{Ann. Statist.}, 40\penalty0 (2):\penalty0 812--831, 2012.

\bibitem[St{\"a}dler et~al.(2010)St{\"a}dler, B{\"u}hlmann, and s~{van de
  Geer}]{Stadler_Buhlmann_vandeGeer10}
N.~St{\"a}dler, P.~B{\"u}hlmann, and Sara s~{van de Geer}.
\newblock {$\ell_1$}-penalization for mixture regression models.
\newblock \emph{TEST}, 19\penalty0 (2):\penalty0 209--256, 2010.

\bibitem[Sun and Zhang(2012)]{Sun_Zhang12}
T.~Sun and C.-H. Zhang.
\newblock Scaled sparse linear regression.
\newblock \emph{Biometrika}, 99\penalty0 (4):\penalty0 879--898, 2012.

\bibitem[Sun and Zhang(2013)]{Sun_Zhang13}
T.~Sun and C.-H. Zhang.
\newblock Sparse matrix inversion with scaled lasso.
\newblock \emph{J. Mach. Learn. Res.}, 14:\penalty0 3385--3418, 2013.

\bibitem[Tibshirani(1996)]{Tibshirani96}
R.~Tibshirani.
\newblock Regression shrinkage and selection via the lasso.
\newblock \emph{J. Roy. Statist. Soc. Ser. B}, 58\penalty0 (1):\penalty0
  267--288, 1996.

\bibitem[Tibshirani and Taylor(2012)]{Tibshirani_Taylor12}
R.~Tibshirani and J.~Taylor.
\newblock Degrees of freedom in lasso problems.
\newblock \emph{Ann. Statist.}, 40\penalty0 (2):\penalty0 1198--1232, 2012.

\bibitem[{van de Geer} and B{\"u}hlmann(2009)]{vandeGeer_Buhlmann09}
S.~{van de Geer} and P.~B{\"u}hlmann.
\newblock On the conditions used to prove oracle results for the {Lasso}.
\newblock \emph{Electron. J. Stat.}, 3:\penalty0 1360--1392, 2009.

\bibitem[Wang et~al.(2015)Wang, Dunson, and Leng]{Wang15}
X.~Wang, D.~Dunson, and C.~Leng.
\newblock No penalty no tears: Least squares in high-dimensional linear models.
\newblock \emph{arXiv:1506.02222}, 2015.

\bibitem[Wasserman and Roeder(2009)]{Wasserman09}
L.~Wasserman and K.~Roeder.
\newblock High dimensional variable selection.
\newblock \emph{Ann. Stat.}, 37\penalty0 (5A):\penalty0 2178, 2009.

\bibitem[Ziegler(1995)]{Ziegler95}
G.~Ziegler.
\newblock \emph{Lectures on polytopes}, volume 152.
\newblock Springer, 1995.

\end{thebibliography}


\appendix
\section{Proofs}\label{sec:proofs}

\subsection{Definitions}

A subtlety in the definition of the lasso variable selection scheme is that it is
defined as the solution to the minimization problem \eqref{def:lasso}, but the solution is not necessarily unique.
Different lasso algorithms can yield different solutions to the lasso problem, and all could reasonably called the lasso
estimator.

Define the lasso equicorrelation set~\cite{Tibshirani_Taylor12} to be
\begin{align*}
E[\lambda;Y]:=\left\{i\in[p]\, : \,|X_i^\top(Y-X\hat\beta^\lambda)|=\lambda\right\}.
\end{align*}
This set is unique and contains the support of any lasso solution $\hat\beta^\lambda$. In common cases, there is at least one lasso solution $\hat\beta^\lambda$ whose support equals the equicorrelation set. This property turns out to be quite valuable in our analysis and consequently, we will always assume from now on that the support set equals the equicorrelation set.

Denote the sets of lasso outputs that lead to the same sign vectors $\eta\in\{1,0,-1\}^p$ by $W^\lambda(\eta) :=\{Y\in\R^n \,:\, \sgn[\hat\beta^\lambda]=\eta\}$, cf.~\cite{Lee15}. The closures of the sets of equal sign vector will be called regions, and the collection of all regions will be written $\mathcal{V}=\{\cl W^1(\eta)\, : \, \eta\subset\{-1,0,1\}^p\}$. To simplify notation, we will write the target as $\xi=X\beta$. When relevant, this will also be written $\xi_n$ to emphasize the dependence on $n$.

\subsection{Lemmas}

Recall that in convex geometry, a polyhedron is a finite intersection of closed half-spaces -- more details can be found in Appendix \ref{sec:proofs}.

\begin{lemma}\label{lem:lasso-regularity}
The lasso $\hat{\poi}$ fulfills the following:
\begin{enumerate}
\item it is scale-symmetric, in the sense that $\supp[\hat\beta^\lambda(Y)] =\supp[ \hat\beta^1(Y/\lambda)]$ for all $\lambda\in\Lambda$, $Y \in \R^n$, and $X\in \R^{n \times p}$;
\item for all $\lambda\in\Lambda$ and $\eta\in\{1,0,-1\}^p$, the closure of its regions of equal sign vector $\cl[\mathcal W^\lambda(\eta)]$ are polyhedra.
\end{enumerate}
\end{lemma}
\begin{proof}[Proof of Lemma \ref{lem:lasso-regularity}]
We prove each condition in order.

\noindent\textit{Part i)}
The scale-symmetry follows from consideration of the dual problem to (\ref{def:lasso}). Let $\hat\beta_\lambda$
be a lasso solution whose active set equals the equicorrelation set $E[\lambda;Y]$.
Let $C_\lambda$ stand for the polyhedron $\{x\in\R^p : \, \|X^\top x\|_\infty\leq\lambda\}$,
and let $P_{C_\lambda}$ denote the Euclidean projection on this set. Notice that for any $x\in C_\lambda$,
we have $x/\lambda\in C_1$ and therefore
\begin{align*}
\|\lambda P_{C_1}\left(\frac{Y}{\lambda}\right)-Y\|_2
\leq
\lambda \|P_{C_1}\left(\frac{Y}{\lambda}\right)-\frac{Y}{\lambda}\|_2
\leq
\lambda \|\frac{x}{\lambda}-\frac{Y}{\lambda}\|_2
\leq
\|x-Y\|_2.
\end{align*}
Since this is true for all $x\in C_\lambda$, and that $\lambda P_{C_1}\left(\frac{Y}{\lambda}\right)\in C_\lambda$,
we conclude that $\lambda P_{C_1}\left(\frac{Y}{\lambda}\right)=P_{C_\lambda}(Y)$.
As shown in \cite{Tibshirani_Taylor12}, the residual from the lasso satisfies
\begin{align*}
Y-X\hat\beta_\lambda=P_{C_\lambda}(Y).
\end{align*}
Therefore, for any $\lambda>0$ the active set of $\hat\beta_\lambda,$ which is the equicorrelation set here, satisfies $E[\lambda;Y]=E\left[1;Y/\lambda\right]$ since
\begin{align*}
E[\lambda;Y]
=\left\{i\in[p] : \big|X_i^\top P_{C_\lambda}(Y)\big|\!=\!\lambda\right\}
=\left\{i\in[p] : \left|X_i^\top P_{C_1}\left(\frac{Y}{\lambda}\right)\right|\!=1\!\right\},
\end{align*}
as desired.

\noindent\textit{Part ii)}
The polyhedron $C_1=\{x\in\rn\,:\,\|X^\top x\|_\infty\leq 1\}$ has an irreducible decomposition into half-spaces
\[C_1\!=\!\left(\bigcap_{i=1}^p\{x\in\rn : X_i^\top x-1\leq\!0\} \right)\bigcap
\left(\bigcap_{i=1}^p\{x\in\rn : -X_i^\top x-1\leq\!0\}\right),\]
so the facets of $C_1$ are $F_{i}^{\pm}=C_1\cap \{x\in\rn : \pm X_i^\top x-1\leq 0\}$ -- see \cite[Sec.\ 2.6]{Grunbaum03}.
Since by assumption, the active set of the lasso estimate coincides with the equicorrelation set,
by the Karush-Kuhn-Tucker conditions and $Y-X\hat\beta_1=P_{C_1}Y$, see \cite[Equations\ (13)-(14) and Lemma\ 3]{Tibshirani_Taylor12}, we have
\begin{align*}
X_i^\top P_{C_1}(Y)=\sgn\hat\beta_i\in\{-1,1\}\;\Leftrightarrow\; P_{C_1}Y\in F_i^{\sgn\hat\beta_i}
\end{align*}
for $i\in E[1;Y]$. Moreover,
\begin{align*}
|X_i^\top P_{C_1}(Y)|<1\;\Leftrightarrow\; P_{C_1}Y\in \big( C_1\backslash F_i^+\big)\cap \big( C_1\backslash F_i^-\big)
\end{align*}
for $i\not\in E[1;Y]$. In light of this, $\sgn\hat\beta=\eta$ if and only if
\begin{align*}
P_{C_1}(Y)&\in
\text{\raisebox{5pt}{$\bigcap_{\substack{i\in [p]\\\eta_i=1}}$}}F_{i}^+
\cap
\text{\raisebox{5pt}{$\bigcap_{\substack{i\in [p]\\\eta_i=\text -1}}$}}F_{i}^-
\cap
\text{\raisebox{5pt}{$\bigcap_{\substack{i\in [p]\\\eta_i\neq 0}}$}}
C_1\backslash F_i^+\cap C_1\backslash F_i^-
\\
&~~~~~=
\relint \Bigg[\Bigg(
\text{\raisebox{5pt}{$\;\bigcap_{\substack{i\in [p]\\\eta_i=1}}$}}
F_{i}^+
\Bigg)\cap\Bigg(
\text{\raisebox{5pt}{$\;\bigcap_{\substack{i\in [p]\\\eta_i=\text -1}}$}}
F_{i}^-\Bigg)\Bigg].
\vspace{-10pt}
\end{align*}
So, $W(\eta )=P_{C_1}^{-1}(\relint F_\eta)$ for the face $F_\eta=\left(\bigcap_{\substack{i\in [p]\\\eta_i=1}}F_{i}^+\!\right)
\cap  \left(\bigcap_{\substack{i\in [p]\\\eta_i=\text -1}}F_{i}^-\!\right)$.

Let $V\in\mathcal{V}$ - then by definition of $\mathcal{V}$, there must be an $\eta \subset\{-1,0,1\}^p$ such that $V=\cl W(\eta)=\cl P_{C_1}^{-1}(\relint F_\eta )$.
According to \cite[Equation~(2.3) and Page~83]{Schneider13},
we have $N(C_1,F_\eta)+\text{relint}\,F_\eta=P_{C_1}^{-1}(\text{relint}\,F_\eta)=W(\eta)$, where $N(C_1,F_\eta)$ is the normal cone of $C_1$ to the face $F_\eta$.

Now, since $\cl A+\cl B\subset \cl(A+B)$ for any sets $A$, $B$,
\begin{align*}
N(C_1,F_\eta)+\text{relint}\,F_\eta\subset N(C_1,F_\eta)+F_\eta\subset\text{cl}\,P_{C_1}^{-1}(\text{relint}\,F_\eta).
\end{align*}
But $N(C_1,F_\eta)+F_\eta$ is the sum of two polyhedra, hence a polyhedron, hence closed. Thus $N(C_1,F_\eta)+F_\eta=\text{cl}\,P_{C_1}^{-1}(\text{relint}\,F_\eta)=V$ is a polyhedron.

Notice that $N(C_1,F_\eta)+\relint F_\eta$ is convex as a sum of convex sets, and thus
\begin{align*}
\relint W(\eta)=&\relint P_C^{-1}(\relint F_\eta)=\relint (N(C_1,F_\eta)+\relint F_\eta)\\
=&\relint\cl(N(C_1,F_\eta)+\relint F_\eta)=\relint V=\intr V,
\end{align*}
since $V$ is $n$-dimensional.
Thus $\intr V\subset W(\eta)$, and $\sgn\hat\beta$ is constant over $\intr V$, as desired.
\end{proof}

\begin{theorem}\label{thm:lepski} For any $t\geq6$, there exists a deterministic integer~ $N$ such that with probability at least $1-t^{-1}-R_n^{-1}\sqrt{1+t\log R_n}$,
\begin{align}
\big\|X\bar\beta^{i,j}-X\beta\big\|_2^2\leq\sigma^2(1+t\log r)|\hat S^{i,j}|
+\left\|(\operatorname{I}_n-P_{\hat S^{i,j}})X\beta\right\|_2^2
\label{eq:UB}
\end{align}
for all $i,j$ and all $n\geq N$.
\end{theorem}
\begin{proof}[Proof of Lemma \ref{thm:lepski}]

By \eqref{eq:model}, the loss of these estimators can be broken down as
\begin{align}&
\left\|X\br^{i,j}-X\beta\right\|_2^2=
\left\|P_{\hat S^{i,j}}(X\br^{i,j}-X\beta)\right\|_2^2
+\left\|(\operatorname{I}_n-P_{\hat S^{i,j}})(X\br^{i,j}-X\beta)\right\|_2^2
\notag\\&\qquad\quad=\;
\left\|P_{\hat S^{i,j}}(XX^+_{\hat S^{i,j}}Y-Y+\varepsilon)\right\|_2^2
+\left\|(\operatorname{I}_n-P_{\hat S^{i,j}})(XX^+_{\hat S^{i,j}}Y-X\beta)\right\|_2^2
\notag\\&\qquad\quad=\;
\left\|P_{\hat S^{i,j}}\varepsilon\right\|_2^2
+\left\|(\operatorname{I}_n-P_{\hat S^{i,j}})X\beta\right\|_2^2,
\label{eq:start-bd}
\end{align}
where $X^+_{\hat S^{i,j}}$ is the Moore-Penrose pseudo-inverse of the submatrix of $X$ comprising the columns with indexes in $\hat S^{i,j}.$ Our goal is to control the noise term $\|P_{\hat S^{i,j}}\varepsilon\|_2^2$ in \eqref{eq:start-bd}. We do this in four parts. We first show that on an appropriate scale, the response $Y$ must be close to the target $X\beta$ with high probability. This is then shown to imply, using the scale-symmetry property of the lasso, that the ordered active sets must be unique with high probability. This allows us to control the noise term by showing that the projected noise behaves like a chi-square distribution, construct an appropriate event, and bound its probability. Finally, on this event the inequality of the theorem is shown to hold.

\begin{enumerate}
\item\label{par:UB-1} We first use a Gaussian tail bound to show that $Y$ is close to the target $X\beta$; precisely, we show that the event
  \begin{equation*}
\bar\Omega_1:=\bigcap_{n=n_1}^\infty\left\{\|Y-{X\beta}\|_\infty\leq \sqrt{6\sigma^2\log n}\right\}
  \end{equation*}
fulfills the bound
\begin{equation}
\mpr\left[\bar\Omega_1\right]\geq 1-1/t,\label{partone}
\end{equation}
where $n_1:=\min \{n \,:\,\allowbreak \sqrt{6\sigma^2\log n}>1/\sqrt{2\pi}\}\vee\lceil 4t\rceil$.\\
For this, write $\xi_n:=X\beta$ (the subscript $n$ highlights the sample size dependence) and define the event $\Omega_n:=\left\{\|Y-\xi_n\|_\infty\leq \sqrt{6\sigma^2\log n}\right\}$ for ease of notation. Using a union bound, the Gaussian tail bound
$\mpr[\text{N}(0,1)>t]\leq e^{-t^2/2}/\sqrt{2\pi} t\allowbreak<e^{-t^2/2}$ (note that $t>6>1/\sqrt{2\pi}$), and  the definition of $n_1$, we find that for the complements $\Omega_n^C$ of the sets $\Omega_n$,
\begin{align*}
\sum_{n=n_1}^\infty\mpr\left[\Omega_n^C\right]&
\leq \sum_{n=n_1}^\infty 2ne^{-3\log n}
=2\sum_{n=n_1}^\infty \frac1{n^2}\\
&\leq \frac{2}{n_1^2}+2\int_{n_1}^\infty\frac1{w^2}dw
=\frac{2}{n_1^2}+\frac{2}{n_1}\leq\frac 4 {n_1}
\leq \frac 1 t.
\end{align*}
From this and the definition of $\bar\Omega_1$, the bound~\eqref{partone} follows.
\item\label{par:UB-2} We now use Part \ref{par:UB-1} to deduce that the active sets are deterministic if the response is close to the target, or more specifically, we derive that on $\bar\Omega_1$,
\begin{equation}\label{parttwo}
\hat S^i[Y]=\hat S^i[\xi_n]
\end{equation}
for $1\leq i\leq r$.\\
From Part \ref{par:UB-1}, we deduce that on $\bar\Omega_1$ and for $n\geq n_1$,
\begin{align*}
\frac{\|Y-\xi_n\|_\infty}{\sqrt{n}}\leq \sqrt{\frac{6\sigma^2\log n}{n}}=\frac{\sqrt{6\sigma^2\log n /n}}{D(\xi_n)}D(\xi_n).
\end{align*}
\noindent Then, by Assumption \ref{ass:app},
there must be an $n_2$, without loss of generality satisfying $n_2\geq n_1$, such that $\|Y-\xi_n\|_\infty/\sqrt{n}<D(\xi_n)$.
But by definition of $D(\xi_n)$, this means that for all $\lambda>0$,
there must be a $\lambda'>0$ such that $\supp[\hat\beta^1(Y/\lambda)]=\supp[\hat\beta^1( \xi_n/\lambda')]$. For a given $y\in\rn$, we now define the collection of active sets by $\hat Q[y]:=\{\supp[\hat\beta^1(y/\lambda)]:\lambda>0\}$. That is, $\hat Q[Y]$ is the collection of active sets along the tuning parameter path of the estimator for given data $(Y,X)$. {There are at most $2^p$ different subsets of $[p]$, so these are always finite sets.} We therefore obtain
\begin{align*}
\mpr\left[\hat Q[Y]=\hat Q[\xi_n]\;\;;\;\;\forall n\geq n_2\;\,:\,\bar\Omega_1\right]=1.
\end{align*}
In particular, the random cardinality $r:=|\hat Q[Y]|$ and deterministic cardinality $\bar r:=|\hat Q[\xi_n]|$ coincide almost surely on this event, that is,
\begin{align*}
\mpr\left[r=\bar r\;\;;\;\;\forall n\geq n_2\;\,:\,\bar\Omega_1\right]=1.
\end{align*}
Because we follow a fixed ordering rule, we must then have \begin{align*}
\mpr&\Big[(\hat S^1[Y],...,\hat S^r[Y])
=(\hat S^1[\xi_n],...,\hat S^r[\xi_n])\;\;\;\;\forall n\geq n_2\,:\,\bar\Omega_1\Big]=1.
\end{align*}
This finishes the proof of Equation~\eqref{parttwo}.
\item\label{par:UB-3} Let us define the sets $\hat S^{i,j}:=\hat S^{i,j}[Y]:=\hat S^i[Y]\cup\hat S^j[Y]$, the random ranks $r^{i,j}:=\rk X_{\hat S^{i,j}}$ and deterministic ranks $\bar r^{i,j}:=\rk X_{\hat S^{i,j}[\xi_n]} $.
Our next step is to show that Part \ref{par:UB-2} provides a chi-square bound for the noise part in~\eqref{eq:start-bd}, that is, we prove on $\bar \Omega_1$ the relations
\begin{align}
    \label{partthree}
\|P_{\hat S^{i,j}}\varepsilon\|_2^2\sim{\sigma^2}\chi^2_{\bar r^{i,j}}\text{~if~} \bar r^{i,j}\geq 1\text{~~and~~}\|P_{\hat S^{i,j}}\varepsilon\|_2^2=0\text{~if ~}\bar r^{i,j}=0.
  \end{align}
To show this, we apply \eqref{eq:start-bd} to our two-step method $\bar\beta^{i,j}$ to get
{\begin{align*}
\big\|X\bar\beta^{i,j}-X\beta\big\|_2^2=\left\|P_{\hat S^{i,j}}\varepsilon\right\|_2^2
+\left\|(\operatorname{I}_n-P_{\hat S^{i,j}})X\beta\right\|_2^2.
\end{align*}}
According to Part \ref{par:UB-2}, the sets $\hat S^{i,j}$ satisfy
\begin{align*}
\mpr\left[\hat S^{i,j}=\hat S^{i,j}[\xi_n]\;\; \;\;\forall n\geq n_2\;\,:\,\bar\Omega_1\right]=1
\end{align*}
for all $1\leq i,j\leq r$. This has two consequences. First, the random ranks $r^{i,j}$ equal the deterministic ranks $\bar r^{i,j}$ almost surely on the event $\bar\Omega_1$:
\begin{align*}
\mpr\Big[r^{i,j}=\bar r^{i,j}\;\;\;\;\forall n\geq n_2\;\,:\,\bar\Omega_1\Big]=1.
\end{align*}
Second, the matrices $P_{\hat S^{i,j}}$ are indeed projection matrices  on $\bar\Omega_1$ of rank $r^{i,j}$, since on this event the active sets $\hat S^{i,j}$ (and, therefore, the subspaces spanned by $X_{\hat S^{i,j}}$) are constant. Formally,
\begin{align*}
\mpr\left[P_{\hat S^{i,j}}=P_{\hat S^{i,j}[\xi_n]}\;\;\;\;\forall n\geq n_2\;\,:\,\bar\Omega_1\right]=1.
\end{align*}
\noindent Combining this with $\varepsilon\sim \mathcal N(0,\sigma^2\operatorname I_n)$ yields the results in~\eqref{partthree}. We can now control the noise part in~\eqref{eq:start-bd} with a chi-square Chernoff bound. More specifically, we obtain the bound
\begin{equation}\label{partfour}
\mpr[\bar\Omega_2]\geq {1-t^{-1}-R_n^{-1}\sqrt{1+t\log R_n}}
\end{equation}
for $\bar\Omega_2:=\bar\Omega_1\cap\left\{\|P_{\hat S^{i,j}}\varepsilon\|_2^2\leq (1+t\log r)\bar r^{i,j}\text{ for all }1\leq i,j\leq r\right\}$. To this end, recall that by definition, the integers $\bar r^{i,j}=|\hat S^{i,j}[\xi_n]|$ and $\bar r=|\hat Q[\xi_n]|$ are deterministic. According to Part \ref{par:UB-2}, it also holds that on $\bar\Omega_1$, we have $\bar r^{i,j}=|\hat S^{i,j}[Y]|$ and $\bar r=r$.  We use this, result \eqref{partthree}, a union bound, and the bound $\mpr[\bar\Omega_1]\geq 1 - 1/t$ stated in \eqref{partone} to deduce
\begin{align*}
\mpr\left[\bar\Omega_2\right]
=&
\mpr\left[\bar\Omega_1 \cap\left\{
\raisebox{10pt}{$\max\limits_{\substack{1\leq i,j\leq \bar r\\\bar r^{i,j}\geq 1}}$}
\frac{\chi^2_{\bar r^{i,j}}}{\bar r^{i,j}}\leq 1+t\log \bar r\right\}\right]
\\
\geq&
1-t^{-1}-\sum_{\substack{1\leq i,j\leq \bar r\\\bar r^{i,j}\geq 1}}
\mpr\Bigg[\chi^2_{\bar r^{i,j}}>\bar r^{i,j}(1+t\log \bar r)\Bigg].
\end{align*}
{Now, using the chi-square Chernoff bound $\mpr[\chi^2_k>k(1+a)]<[(1+a)e^{-a}]^{k/2}$, we obtain
\begin{align*}
\mpr\left[\bar\Omega_2\right]&\geq\;
1-t^{-1}-\sum\limits_{\substack{1\leq i,j\leq \bar r\\\bar r^{i,j}\geq 1}}
\left[(1+t\log \bar r)e^{-t\log\bar r}\right]^{\bar r^{i,j}/2}.
\end{align*}
As $\bar r\geq1$, $(1+t\log \bar r)e^{-t\log\bar r}\leq1$ and so
\begin{align*}
\mpr\left[\bar\Omega_2\right]&\geq\;
1-t^{-1}-\sum\limits_{\substack{1\leq i,j\leq \bar r\\\bar r^{i,j}\geq 1}}
\left[(1+t\log \bar r)e^{-t\log\bar r}\right]^{1/2}\\
&\geq\;
1-t^{-1}-\bar r^{2-t/2}\sqrt{1+t\log \bar r}.
\end{align*}
We now use that that $2-t/2\leq-1$ for all $t\geq6$  and the fact that $\bar r\geq R_n$ to find that
\begin{align*}
\mpr\left[\bar\Omega_2\right]\geq 1-t^{-1}-R_n^{-1}\sqrt{1+t\log R_n},
\end{align*}
which concludes Part \ref{par:UB-3}.}
\item\label{par:UB-4} We finally collect the pieces to deduce that with probability at least {$1-t^{-1}-R_n^{-1}\sqrt{1+t\log R_n}$}, the bound
\begin{align}
\big\|X\bar\beta^{i,j}-X\beta\big\|_2^2\leq\sigma^2(1+t\log r)|\hat S^{i,j}|
+\left\|(\operatorname{I}_n-P_{\hat S^{i,j}})X\beta\right\|_2^2
\label{partfive}
\end{align}
holds for all $n\geq n_2$.\\
For this, we assume that indeed $n\geq n_2$ and then combine the initial bound~\eqref{eq:start-bd} and the results of Part \ref{par:UB-3} to find that on $\bar\Omega_2$,
\begin{align*}
\big\|X\bar\beta^{i,j}-X\beta\big\|_2^2
&\leq\sigma^2(1+t\log r)\bar r^{i,j}
+\left\|(\operatorname{I}_n-P_{\hat S^{i,j}})X\beta\right\|_2^2.
\end{align*}
Recalling that $\bar r^{i,j}=|\hat S^{i,j}|$ according to Part \ref{par:UB-2}, the desired bound \eqref{partfive} now follows from Inequality~\eqref{partfour} derived in Part \ref{par:UB-3}.
\end{enumerate}
\end{proof}

\subsection{Proofs of Results From Section 2}

Recall that the path of active sets is $|\hat S^1|\leq|\hat S^2|\leq \dots\leq |\hat S^r|$. The cardinality $r$ is typically random, but we can always bound it almost surely by some constant $\boundforparameter$, so that $1\leq \boundforparameter\leq r$ almost surely. This constant should be independent of the data but can be chosen to vary with $n$ and $p$. For example, we might have agreed \textit{a priori} with considering $50$ sets, or our variable selection method might always yield at least $\min(n,p)$ different sets by construction.

\begin{corollary}[Oracle benchmark]\label{cor:oracle-benchmark} Say the oracle set exists, that the design satisfies Assumption~\ref{ass:app}.  Then, for any $t\geq6$, there exists a deterministic integer $N$ such that with probability at least $1-t^{-1}-\boundforparameter^{-1}\sqrt{1+t\log \boundforparameter}$, the oracle estimator satisfies
\begin{align*}
\|X\obe-X\beta\|_2^2\;\leq\;&\sigma^2(1+t\log r)|\ose|
\end{align*}
for all $n\geq N$.
\end{corollary}
\begin{proof}[Proof of Corollary \ref{cor:oracle-benchmark}]
The oracle estimator is $\beta^*=\br^{i*}$, the refitted estimator on the oracle set $S^*=\hat S^{i^*}$. The result therefore follows immediately from Theorem \ref{thm:lepski} applied to $i=j=i^*$.
\end{proof}

\begin{corollary}[Oracle inequality for \avt]\label{cor:oracle-inequality} Say the oracle set exists, that the design satisfies Assumption~\ref{ass:app}, and that the $AV_p$ parameter $a$ is such that $a\geq 2\sigma^2(1+t\log r)$.  Then, for any $t\geq6$, there exists a deterministic integer $N$ such that with probability at least $1-t^{-1}-\boundforparameter^{-1}\sqrt{1+t\log \boundforparameter}$, it holds that $|\hat S |\leq |\ose|$ and
\begin{align*}
  \|X\br-X\beta\|_2^2\;\leq\;&\Big[6a+4\sigma^2(1+t\log r)\Big]|\ose|\eqsp
\end{align*}
for all $n\geq N$.
\end{corollary}
\begin{proof}[Proof of Corollary \ref{cor:oracle-inequality}]

By Theorem \ref{thm:lepski}, there exists an $N$ such that the event
\begin{align*}
\Omega=&\bigg\{
\big\|X\bar\beta^{i,j}-X\beta\big\|_2^2\leq\sigma^2(1+t\log r)|\hat S^{i,j}|\\
&~~~~~+\left\|(\operatorname{I}_n-P_{\hat S^{i,j}})X\beta\right\|_2^2, \forall i,j\text{ and }n\geq N\bigg\}
\end{align*}
holds with probability at least $1-t^{-1}-R_n^{-1}\sqrt{1+t\log R_n}$.

~\\\noindent\textit{Claim (i): On $\Omega$, it holds that $\ein\leq \oin=\min\big\{{i\in[r]}\big|\hat S^i\supset S\big\}$.}\\
We prove this claim by contradiction and, therefore, assume that $\ein > \oin$. Then, by the definition of our estimator, there must be an integer $k\in[r]$ such that $|\hat S^k|\geq|\ose|$ and
\begin{equation}
  \|X{\obe}-X\overline{\beta}^{k , \oin}\|_2^2
   > a |\ose|+a|\hat S^k\cup \ose|\eqsp.\label{eq:opt-1}
\end{equation}
The fact that $|\hat S^k|\geq|\ose|\geq|S|$, together with the bound~\eqref{eq:UB} and $\ose\supset S$, yields
\begin{align*}&
\|X{\obe}-X\overline{\beta}^{k, \oin} \|_2^2
\\\leq&\;
2\|X\obe-X\beta\|_2^2+2\|X\beta-X\br^{k , \oin} \|_2^2
\\ \leq&\;
2\sigma^2(1+t\log r)|\ose|+2\|(\operatorname{I}_n-P_{\ose})X\beta\|_2^2+2\sigma^2(1+t\log r)|\hat S^k\cup \ose|\\
&+2\|(\operatorname{I}_n-P_{\hat S^k\cup \ose})X\beta\|_2^2
\\ \leq &\;
2\sigma^2(1\!+\!t\log r)|\ose|+2\|(\operatorname{I}_n\!-P_{S})X\beta\|_2^2+2c|\hat S^k\cup\ose|+2\|(\operatorname{I}_n\!-P_{S})X\beta\|_2^2
\\ =&\; 2\sigma^2(1+t\log r)|\ose|+0+2\sigma^2(1+t\log r)|\hat S^k\cup \ose|+0 \eqsp.
\end{align*}
Since $a\geq 2\sigma^2(1+t\log r)$, this contradicts (\ref{eq:opt-1}) and, therefore, concludes the proof of Claim (i).

~\\\noindent\textit{Claim (ii): On $\Omega$, it holds that $\|X\br-X\beta\|_2^2\leq (6a+4\sigma^2(1+t\log r))|\ose|$.}\\
To prove this claim, we note that by Claim 1, we have $\ein \leq \oin$ and, therefore, $|\se|\leq |\ose|$. Hence, the definition of the estimator implies for $\ein = r$ that
\begin{equation*}
  \|X\br-X\br^{\ein,\oin}\|_2^2=\|X\br-X\br^{\ein,\ein}\|_2^2=0
\end{equation*}
and otherwise, if $\ein < r$, that (recall that $\ose\supset S$)
\begin{equation*}
  \|X\br-X\br^{\ein,\oin}\|_2^2\leq a |\se|+a |\se \cup \ose|\leq 3a|\ose|\eqsp.
\end{equation*}
The bound~\eqref{eq:UB}, on the other hand, yields
\begin{equation*}
    \|X\br^{\ein,\oin}-X\beta\|_2^2\leq \sigma^2(1+t\log r)|\se \cup\ose|+0\leq 2\sigma^2(1+t\log r)|\ose|\eqsp.
\end{equation*}
Combining these two inequalities, we finally obtain
 \begin{align*}
    \|X\br-X\beta\|_2^2&\leq 2 \|X\br-X\br^{\ein,\oin}\|_2^2+2 \|X\br^{\ein,\oin}-X\beta\|_2^2\\
&\leq (6a+4\sigma^2(1+t\log r))|\ose|\eqsp,
  \end{align*}
which concludes the proof of Claim (ii).
\end{proof}

In particular, this yields the results of Section \ref{sec:lepski}.

\begin{proof}[Proof of Proposition \ref{prop:oracle-benchmarkshort}]
Let $t=\max(\frac2{\alpha},6)$, and let $R$ be large enough so that \[t^{-1}+R^{-1}\sqrt{1+t\log R}\;\leq\; \frac{\alpha}2+\sqrt{\frac1{R^2}+\frac2{\alpha}\frac{\log R}{R^2}}\;\leq\; \alpha.\]
Using Corollary \ref{cor:oracle-benchmark} with $t$ and $R_n=R$ gives the result.
\end{proof}

\begin{proof}[Proof of Theorem \ref{thm:oracle-inequalityshort}]
Let $t=\max(\frac2{\alpha},6)$, and let $R$ be large enough so that \[t^{-1}+R^{-1}\sqrt{1+t\log R}\;\leq\; \frac{\alpha}2+\sqrt{\frac1{R^2}+\frac2{\alpha}\frac{\log R}{R^2}}\;\leq\; \alpha.\]
Using Corollary \ref{cor:oracle-inequality} with $t$ and $R_n=R$ gives the result.
\end{proof}

\subsection{Proof of Theorem \ref{thm:lasso-dpositive}}

The proof of Theorem \ref{thm:lasso-dpositive} rely on various convex geometry notions, so we first remind the reader of some background on the subject.

The affine hull of a set $A$, denoted $\text{aff}\,A$ is the intersection of all affine spaces that contain $A$, or alternatively, the unique affine set of minimal dimension that contains $A$.
We write, by extension, $\text{dim}\,A=\text{dim}\,\text{aff}\,A$. We denote the interior, closure, and boundary of $A$ by $\text{int}\,A$, $\text{cl}\,A$, and $\partial A$, respectively.
The relative interior and boundary of $A$, denoted  $\text{relint}\,A$ and $\text{relbd}\,A$, are respectively the
interior and the boundary when $A$ is seen as a subset of its affine hull. We write $A\subset B$ if $A$ is a (not necessarily strict) subset of $B$.

A half-space $H^+$ is a set of the form $\{x\in\R^n\,:\,\alpha^\top x\leq b\}$ for $\alpha\in \R^n$, $b\in \R$.
Its boundary $H=\partial H^+=\{x\in\R^n\,:\,\alpha^\top x=b\}$ is a hyperplane in $\R^n$.
A polyhedron is a finite intersection of half-spaces, $X=\bigcap_{i\in I} H^+_i$.
 Such decompositions are usually not unique;
we call a decomposition irreducible if $\bigcap_{j\not=i} H^+_j\not=X$ for all $i\in I$. Given an irreducible decomposition, a facet of $X$ is a set $F_i=X\cap H_i$. The faces are the intersections of (potentially many) facets.
The normal cone to a point $x_0\in X$ is $N(X,x)=\{y\,:\,y^\top(x-x_0)\leq0\text{ for all }x\in X\}$. One can show \cite[Page\ 83]{Schneider13} that for a given face $F$, all $x_0\in F$ have the same normal cone; hence we define the normal cone to $F$ to be $N(X,F)=N(X,x_0)$ for any $x_0\in F$.

Next, we make the following remarks. Recall that $\mathcal{V}$ is the collection of regions, namely the closures of sets of points that have the same sign vector under the lasso.

By Lemma \ref{lem:lasso-regularity}, then the regions must have disjoint interiors. Indeed, for $V\not=V'\in\mathcal{V}$ we must have $V=\cl W^1(\eta)$, $V'=\cl W^1(\eta')$ for some $\eta\not=\eta'\subset\{-1,0,1\}^p$.
Since $\sgn\hat\beta$ is constant on $\intr V$, we conclude that $\intr V\subset W^1(\eta)$, $\intr V^{\prime}\subset W^1(\eta')$. But $W^1(\eta)\cap W(\eta')=\{Y\,:\, \eta=\sgn\hat\beta^1=\eta'\}=\emptyset$, so $\intr V\cap V^{\prime}=\emptyset$.

Take $V\not=V'\in\mathcal{V}$ again: being polyhedra, their boundaries $\partial V$, $\partial V'$ can be partitioned by the relative interiors of their proper faces \cite[Theorem 2.1.2]{Schneider13}. Let $\mathcal{F}(V)$ denote the set of proper faces of a polyhedron $V$ and $\mathcal{C}=\{\text{relint}\,F\cap \text{relint}\,F'\,:\, F\in\mathcal{F}(V),F'\in\mathcal{F}(V'), V\not=V'\in\mathcal{V}\}$ be the collection of ``boundary pieces''. We enumerate, for reference, two properties of $\cal C$:
\begin{enumerate}[label=\roman*)]
\item For two distinct $V,V'\in\mathcal{V}$ and an $x\in\partial V\cap\partial V'$, there is a unique $B\in\mathcal{C}$ such that $x\in B$.
\item Each $B\in\mathcal{C}$ has dimension at most $n-1$.
\end{enumerate}
Indeed, for the first statement we notice that by partitioning, there exists unique faces $F\in\mathcal{F}(V)$ and $F'\in\mathcal{F}(V')$ such that $x\in\text{relint}\,F\cap\text{relint}\,F'$, hence a unique $B\in\mathcal{C}$ such that $x\in B$. For the second,
since the interiors of $V,V'$ are disjoint, for $B=\text{relint}\,F\cap \,\text{relint}\,F'\not=\emptyset$ we must have $F\in\partial V$,$F'\in\partial V'$, hence $\text{dim}\,B\leq \text{dim}\,\text{relint}\,F\wedge\text{dim}\,\text{relint}\,F'\leq n-1$. In addition to these observations, we will need the following two lemmas and one supporting proposition.

\begin{lemma}\label{lem:lasso-dpositive-l1}
If $B\in\mathcal{C}$ is of dimension at most $n-2$, then $\R_+B$ is of dimension at most $n-1$.
\end{lemma}
\begin{proof}
Let $S=\text{aff}\,B$ be the affine hull of $B$, of dimension at most $n-2$. Being an affine subspace, it can be written $S=\{x\,:\,Ax+b=0\}$ for some  matrix $A\in\R^{n\times n}$ with $\rank A\leq n-2$ and some vector $b\in\R^n$. Now, $B\subset S$ implies $\R_+B\subset \R S$, and $\R S=\{x\,:\,\exists t \,\text{s.t.}\,Atx+b=0\}=\left[I_n\,0\right]\{(x,t)\,:\,Ax+tb=0\}=\left[I_n\,0\right]\text{Ker}[A\,b]$.
Now, since $\rank A \leq n-2$, $[A\,b]$ has rank at most $n-1$, and thus $\text{Ker}[A\,b]$ is a subspace of $\R^{n+1}$ of rank at most $n-1$. Hence, $S=\left[I_n\,0\right]\text{Ker}[A\,b]$ has dimension at most $n-1$.
Since $\R_+B\subset \R S$, and $\R S$ is affine (actually, a subspace), then $\text{aff}\,\R_+B\subset \R S$. Consequently, $\text{dim}\,\R_+B\leq \text{dim}\,\R S\leq n-1$, as desired.
\end{proof}

\begin{proposition}\label{prop:hyperplane} Let $K\subset\R^n$ be a polyhedron of full dimension $n$ with irreducible decomposition $K=\bigcap_j H^+_j$. Denote the hyperplanes by $H_i=\partial H_i^+$ as usual. Then for any facet $F_i=K\cap H_i$ of $K$ and any point $x\in\relint F_i$, there exists an $\varepsilon>0$ such that $\intr H_i^+\cap B(x,\varepsilon)\subset\intr K$, $H_i\cap B(x,\varepsilon)\subset F_i$ and $H_i^{+C}\cap B(x,\varepsilon)\subset K^C$.
\end{proposition}
\begin{proof}
With the irreducible decomposition, $x\in K\cap H_j^+$ for all $j$; say for $j\not= i$ we had $x\in F_j=K\cap H_j$. As argued in \cite[Page\ 27]{Grunbaum03}, $x\in F_i\cap F_j$ implies that $x$ is in a facet of $F_i$, hence in $\partial F_i$, a contradiction with $x\in\relint F_i$. Thus $x\in K\cap \intr H_j^{+}$ for all $j\not= i$.

Since $x\in\bigcap_{j\not= i}\intr H^+_j$, we can find $\varepsilon>0$ such that $B(x,\varepsilon)\subset \bigcap_{j\not= i}\intr H^+_j$. Then $H^+_i\cap B(x,\varepsilon)\subset \bigcap_{j}\intr H^+_j \subset\intr K$, $H_i\cap B(x,\varepsilon)\subset H_i\cap \bigcap_{j\not= i}\intr H^+_j\subset H_i\cap K=F_i$, and finally $H^{+C}_i\cap B(x,\varepsilon)\subset \bigcup_{j}H^{+C}_j=K^C$, as desired.
\end{proof}

\begin{lemma}\label{lem:lasso-dpositive-l2}
Let $B\in\mathcal{C}$ be of dimension $n-1$ and $L$ be a ray centered at the origin such that $B\cap L\not=\emptyset$. Then either $\R_+B$ has dimension $n-1$, or $L\cap \intr V\not=\emptyset$ and $L\cap \intr V^{\prime}\not=\emptyset$.
\end{lemma}
\begin{proof}
Write $B=\text{relint}\,F\cap\text{relint}\,F'$. Since $\text{relint}\,F$, $\text{relint}\,F'$ both have dimension at most $n-1$, and $B$ has dimension $n-1$, then they must have exactly dimension $n-1$. They are thus facets of their respective polyhedra $V,V'\in\mathcal{V}$, which must have dimension $n$, and have exactly one supporting hyperplane. Thus $\text{aff}\,B =\text{aff}\,F =\text{aff}\,F'$, which we might denote $S$. Let $b\in L\cap S$, write $S=\langle a_1,...,a_{n-1}\rangle+b$ for some linearly independent vectors $a_1,...,a_{n-1}$, and let $a_n$ be such that $L=\R_+a_n$. Consider the affine transformation
\begin{align*}
\phi(x)=\left[a_1,...,a_{n-1},a_n\right]x+b=Ax+b,
\end{align*}
which maps vectors $(x_1,...,x_{n-1},0)$ bijectively to $S$, and vectors\linebreak $(0,...,0,x_n)$ with $x_n\geq-\|b\|/\|a_n\|$ bijectively to $L$. (Recall that $b\in S\cap L$.) There are then two possibilities.

\textit{Case i)} A has rank $n-1$. Then $a_n\subset \langle a_1,...a_{n-1}\rangle$ and $L\subset S$. But this means that $0\in L\subset S$, that is, that $S$ is a subspace. Then $\R_+B\subset \R S=S$, that is, $\R_+B$ has dimension $n-1$.

\textit{Case ii)} A has rank $n$. Then $\phi$ is bijective and $L\cap S$ is the singleton $\{b\}$. Recall that $V$ and $V'$ are polyhedra of dimension $n,$ and let $V=\bigcap_i H_i^+$, $V'=\bigcap_i H_i^{\prime +}$ be irreducible representations into half-spaces $H_i^+$, $H_j^{\prime+}$ with boundary $H_i$, $H_j^\prime$. By irreducibility, there are unique indices $i,j$ such that $\text{aff}\,F=H_i=\text{aff}\,F'=H'_j=S$.
Since $b\in\relint F \cap\relint F'$, by Proposition \ref{prop:hyperplane} there must be an $\varepsilon>0$ small enough so that $\intr H^+_i\cap B(b,\varepsilon)\subset \intr V$, $H_i^{+C}\cap B(b,\varepsilon)\subset V^C$, $\intr H_j^{\prime}\cap B(b,\varepsilon)\subset \intr V^{\prime}$ and $H_j^{\prime+C}\cap B(b,\varepsilon)\subset V^{\prime C}$.

But clearly $H_i^+\not=H_j^{\prime+}$, as otherwise $\emptyset\not=B(b,\varepsilon)\cap \intr H_i^+=B(b,\varepsilon)\cap
\intr H_j^{\prime+}\subset \intr V\cap \intr V^{\prime}=\emptyset$, a contradiction since interiors of regions are disjoint.
Thus it must hold that $\intr H_i^+=H_j^{\prime+C}$. In light of this, we may simplify the notation by writing $S^+=\intr H_i^+$ and $S^-=H_j^{\prime+C}$.
The state of affairs is then that $B(b,\varepsilon)\cap S^+\subset \intr V$, $B(b,\varepsilon)\cap S\subset B$ and $B(b,\varepsilon)\cap S^-\subset \intr V^{\prime}$.

Write $\R^{(n-1)+}=\{x\in\R^n\,:\, x_n>0\}$ and $\R^{(n-1)-}=\{x\in\R^n\,:\, x_n<0\}$. Since $\phi$ is open, it must map connected components to connected components, and being bijective it must hold that $\phi(\R^{(n-1)+})=S^+$ and $\phi(\R^{(n-1)-})$ $=S^-$, or vice versa. Fix the former by considering $x\mapsto-Ax+b$ instead of $x\mapsto Ax+b$ if necessary. Since $\phi^{-1}(B(b,\varepsilon))$  is an open set, there must be an $\varepsilon'$ such that $B(0,\varepsilon')\subset \phi^{-1}(B(b,\varepsilon))$.

Let $\varepsilon''=\varepsilon'\wedge\|b\|_2/\|a_n\|_2$, so that $|t|<\varepsilon''$ implies $ta_n+b\in L$. Then, we have that $\emptyset\not=\phi(\{(0,...,0,t),t\in(0,\varepsilon'')\})\subset S^+\cap L\cap B(b,\varepsilon) \subset L\cap \intr V$ and $\emptyset\not=\phi(\{(0,...,0,t),t\in(-\varepsilon'',0)\}) \subset S^-\cap L\cap B(b,\varepsilon)\subset L\cap \intr V^{\prime}$. Thus, $L\cap \intr V$ and $L\cap \intr V^{\prime}$ are non-empty, as desired.
\end{proof}

We may now turn to the proof of the theorem.

\begin{proof}[Proof of Theorem \ref{thm:lasso-dpositive}]
Every region $V\in\cal V$ is a polyhedron, so has a decomposition $V=P(V)+C(V)$
 into a polytope $P(V)$ and a cone $C(V)$ by Minkowski's theorem \cite[Theorem 1.2]{Ziegler95}.
Define $\mathcal{V}_0=\{V\in\mathcal{V}\,:\,\intr V\cap\R_+\xi\not=\emptyset\}$, $T=\bigcup_{V\in\mathcal{V}_0}V$ and
\begin{align*}
R=\bigcup_{\substack{B\in\mathcal{C}\\ \text{dim}\,\R_+B\\\,\leq n-1}}\hspace{-10pt}\R_+B
\quad\cup
\bigcup_{\substack{V\in\mathcal{V}\\\text{dim}\,C(V)\\\leq n-1}}\hspace{-10pt}C(V)
\quad\cup
\bigcup_{\substack{V\in\mathcal{V}}}\partial C(V).
\end{align*}
The set $R$ is a finite union of closed sets of dimension at most $n-1$,
so is closed and has measure zero.
We argue that if $\xi\in R^C$, then $\R_+\xi\subset \intr T$. Indeed, say that $t\xi\in\partial T$ for some $t>0$.
Since $\partial T\subset \bigcup_{V\in\mathcal{V}_0} \partial V$,
there is a $V_0\in\mathcal{V}_0$ such that $t\xi\in \partial V_0$.
Since $\R^n=\bigcup_{V\in\mathcal{V}}V$ and $t\xi\in\partial T$, we have
$t\xi\in \cl(T^C)=\bigcup_{V\not\in\mathcal{V}_0}V$. Thus there must also be a $V_1\not=V_0$, $V_1\not\in\mathcal{V}_0$ such that $t\xi\in V_1$. But since the interiors are disjoint, if $t\xi\in \intr V_1$ there would be a contradiction with $t\xi\in\partial V_0$; hence $t\xi\in\partial V_1$. Thus $t\xi\in \partial V_0\cap \partial V_1$ and there must be a unique $B\in\mathcal{C}$ such that $t\xi\in B$. That piece, like all elements of $\mathcal{C}$, must be of dimension $n-1$ or lower.

We argue that $B\subset R$. If it has dimension $n-2$, then $\text{dim}\,\R_+B\leq n-1$ by Lemma \ref{lem:lasso-dpositive-l1}, so $B$ is indeed a subset of $R$.
Now say instead it has dimension $n-1$ and recall that $t\xi\in\R_+\xi\cap B$. Let $F\in\mathcal{F}(V_0)$ and $F'\in\mathcal{F}(V_1)$ be such that $B=\text{relint}\,F\cap \text{relint}\,F'$.
By Lemma \ref{lem:lasso-dpositive-l2}, we must have either
$\text{dim}\,\R_+B=n-1$, or $\R_+\xi\cap \intr V_0\not=\emptyset$ and $\R_+\xi\cap \intr V_1\not=\emptyset$.
 But if the latter was the case, then $V_1\subset T$ by definition, which is impossible; thus we must have $\text{dim}\,\R_+B=n-1$, hence $B\subset R$ again.

Thus in all cases, $B\subset R$. Let $d(x,A):=\inf_{y\in A}\|x-y\|_2$ denote the Euclidean distance between a point $x$ and a set $A$. Since $t\xi\in B$, we have $d(t\xi,R)=0$. But at the same time, since $\xi\in R^C$ and $R$ is closed we must have $d(\xi,R)>0$, and since $R$ is invariant under positive multiplication,
\begin{align*}
d(t\xi,R)=\inf_{y\in R}\|t\xi\!-\!y\|_2=t\inf_{y/t\in R}\|\xi\!-\!y\|_2=t\inf_{y\in R}\|\xi\!-\!y\|_2=td(\xi,R)>0.
\end{align*}
This is a contradiction, and we conclude that $\R_+\xi\subset \intr T$.

Now, $\mathcal{V}_0$ is finite, since $C$ has only a finite number of faces.
Let $h$ be the continuous map $t\mapsto t\xi$, and consider for each $V\in\mathcal{V}_0$ the closed set $h^{-1}(V)$. This set must be convex, since for $s,t\in h^{-1}(V)$, $h(\gamma s+(1-\gamma)t)=[\gamma s+(1-\gamma)t]\xi=\gamma[s\xi]+(1-\gamma)[t\xi]\in V$ by convexity of $V$ for any $\gamma\in[0,1]$. The only closed, convex sets of $\R$ are the closed intervals: thus $h^{-1}(V)=[\alpha ,\beta]$ for some $\alpha \leq \beta$.

Enumerate arbitrarily the $V\in\mathcal{V}_0$ as $V_1,...,V_m$, and for $V_i\in\mathcal{V}_0$ let $[\alpha _i,\beta_i]=h^{-1}(V_i)$. Now, $\R_+\xi\subset T=\bigcup_{i=1}^m V_i$, so $h^{-1}(T)=\bigcup_{i=1}^m[\alpha _i,\beta_i]=\R_+$. Then some $\beta_i$ must equal $\infty$, otherwise the union would be bounded. Moreover, since the interiors of the $V$'s are disjoint, $(\alpha _i,\beta_i)\cap(\alpha _j,\beta_j)=\emptyset$ for $i\not=j$ and the $\beta_i=\infty$ must be unique, all the others finite. By reordering if necessary, take $0=\alpha _1<\beta_1\leq \alpha _2<\beta_2\leq...\leq \alpha _m<\beta_m=\infty$.

The region $V_m$ is a polyhedron, so has a decomposition as $V_m=P(V_m)+C(V_m)$ for some polytope $P(V_m)$ and cone $C(V_m)$ by
Minkowski's theorem \cite[Theorem 1.2]{Ziegler95}. Fix a point $t_0\in(\alpha_m,\infty)$; then since $t_0\xi+\R_+\xi=(t_0,\infty)\xi\subset \intr V_m$, by \cite[2.5.1]{Grunbaum03}, we conclude that $\R_+\xi\subset C(V_m)$, so $t_0\xi,\xi\in C(V_m)$. If $C(V_m)$ has dimension at most $n-1$, or $t_0\xi\in\partial C(V_m)\Leftrightarrow\xi\in\partial C(V_m)$, then $\xi\in  R$, which contradicts $\xi\in R^C$ -- thus $C(V_m)$ must have dimension
$n$ and $t_0\xi,\xi\in \intr C(V_m)$. Let $\varepsilon_1$ be small enough so that $B(\xi,\varepsilon_1)\subset \intr C(V_m)$.
Then for any $s>0$, $B([t_0+s]\xi,s\varepsilon_1)=t_0\xi+sB(\xi,\varepsilon_1)\subset \intr C_m\subset \intr V_m$.
Thus for all $s>2t_0\|\xi\|_2/\varepsilon_1$, $B(s\xi,s\varepsilon_1/2)\subset B([t_0+s]\xi,s\varepsilon_1)\subset \intr V_m\subset \intr T$.

Next, notice that the segment $[0,2t_0\|\xi\|_2/\varepsilon_1]\xi$ is compact and in $\intr T$. This implies that $d([0,2t_0\|\xi\|_2/\varepsilon_1]\xi,\intr T^C)$ is strictly positive and also that there must be an $\varepsilon_2$-neighborhood such that $B([0,2t_0\|\xi\|_2/\varepsilon_1]\xi,\varepsilon_2)\subset \intr T$. Hence, for all $s\in[0,2t_0\|\xi\|_2/\varepsilon_1]$, it holds $B(s\xi,s\varepsilon_2\varepsilon_1/2t_0\|\xi\|_2)\subset B(s\xi,\varepsilon_2)\subset \intr T$.

Finally, let $\varepsilon=\min(\varepsilon_1/2,\varepsilon_1\varepsilon_2/2t_0\|\xi\|_2)$.
Then for all $s\geq0$, we have $B(s\xi,s\varepsilon)\subset \intr T$. Let $|y-\xi|<\varepsilon$ and define  $\eta=\sgn\hat\beta(sy)$.
Then $\cl W(\eta)=V_0$ for some $V_0\in\mathcal{V}_0$, since otherwise $\cl W(\eta)\subset \cl T^C=\bigcup_{V\not\in\mathcal{V}_0}V$, which would contradict $sy\in \intr T$. But there is a $t>0$ such that $t\xi\in \intr V_0$, since $V_0\in\mathcal{V}_0$, and since $\sgn\hat\beta$ is constant over $\intr V_0$, $\sgn\hat\beta(t\xi)=\sgn\hat\beta(sy)$. Thus in particular $\hat S[sy]=\hat S[t\xi]$. Since this is true for all $|y-\xi|<\varepsilon$, we conclude that $D(\xi)\geq\varepsilon>0$, as desired.
\end{proof}

\section{Description of the lslassoBIC}
\label{sec:lassobic}
In this section, we provide details for the lslassoBIC implementation that we have used. This method is similar to applying a BIC procedure over the refitted models obtained by a lasso path, and was recently analyzed in \cite{Bellec16}.

We consider the same collection of tuning parameters
$\Lambda=\{\lambda_1,\dots,\lambda_r\}$ as before, and we denote the associated supports by $(\hat S^1,\dots,\hat S^r)$, where $\hat S^i:=\supp[\hat\beta^{\lambda_i}]$. Following Equation~\eqref{eq:refitting}, we write $(\bar\beta^{\lambda_1},\dots,\bar\beta^{\lambda_r})$ for the estimated least-squares over these supports.
Let us introduce for each $j\in[r]$ a prior $\pi_{j}$ via
\begin{align*}
    \pi_{j} = \left(H_p \binom{p}{|\hat S^j|} \exp(|\hat S^j|)\right)^{-1}\enspace
\end{align*}
with $H_p=(e-e^{-p})/(e-1)$. Then, the lslassoBIC is defined  by
\begin{align*}
\bar\beta^{\rm lslassoBIC} =\bar\beta^{\lambda_{j^{\star}}} \quad \text{ with }
 \quad    j^{\star} \in \argmin_{j \in [r]} \left( \|Y-X\bar\beta^{\lambda_j}\|_2^2+14 \hat\sigma^2 \log \frac{1}{\pi_{j}} \right) \enspace,
\end{align*}
where $\hat\sigma$ is a standard deviation estimate of the (Gaussian) noise. As the practical estimation of $\sigma$ is not discussed further in \cite{Bellec16}, we have used the same estimator as for \avt, namely the one defined in Algorithm~\ref{algosqrt}.

Note that we have adapted the method proposed by \cite{Bellec16} to the case of a predetermined number of lasso parameters. This is because the number of kinks over the lasso path can be as large as $(3^p+1)/2$ \cite{Mairal_Yu12}, making an estimator based on the entire collection of kinks intractable.

\end{document}